\pdfoutput=1
\documentclass[acmsmall,screen]{acmart}

\usepackage[utf8]{inputenc}
\usepackage[T1]{fontenc}
\usepackage{microtype}
\usepackage{graphicx} %
\usepackage[colorinlistoftodos]{todonotes}
\usepackage{mathpartir}
\usepackage[frozencache=true]{minted}
\usepackage{subcaption}
\usepackage[normalem]{ulem}
\usepackage{cancel}
\usepackage{esvect}
\usepackage{tensor}
\usepackage{float}
\usepackage{mathtools}
\usepackage{MnSymbol}
\usepackage{array}
\usepackage{booktabs}
\usepackage{multirow}
\usepackage{varwidth}

\newcommand{\add}[1]{\textcolor{green}{#1}}
\newcommand{\del}[1]{\textcolor{red}{#1}}

\renewcommand{\add}[1]{#1}
\renewcommand{\del}[1]{}

\newmintinline[pyinline]{python}{}
\newmintinline[ttinline]{text}{}

\newcommand{\kw}[1]{\mbox{\pyinline{#1}}}
\newcommand{\kwtt}[1]{\mbox{\ttinline{#1}}}
\newcommand{\mquot}[1]{\left(#1\right)}

\newcommand{\calcname}{$\lambda^O$}
\newcommand{\toolname}{\textsc{Opal}}
\newcommand{\toolacronym}{\toolname{}\footnote{\toolname{} stands for OPportunistic evALuation.}}

\newenvironment{inlinebox}{
    \vspace{\baselineskip}
    \begin{center}
    \footnotesize
}{
    \end{center}
    \vspace{\baselineskip}
}

\newcommand{\loopbench}{CityExcursions}
\newcommand{\pipelinebench}{TextToSpeech}
\newcommand{\toolusebench}{FactCheck}
\newcommand{\totbench}{TreeofThoughts}
\newcommand{\jsonbench}{JSONDec}

\newcommand{\efun}[2]{\kw{fun}\ {#1} \kw{:}\ {#2}}
\newcommand{\eapp}[2]{{#1}\ {#2}}

\newcommand{\etuple}[1]{\kw{(}{#1}\kw{)}}
\newcommand{\eproj}[2]{\kw{prj}\ {#1}\ {#2}}
\newcommand{\eprim}[1]{\kw{<}\,{#1}\,\kw{>}}
\newcommand{\etask}[2]{\kw{<<}#1 \ #2\kw{>>}}

\newcommand{\sep}{~\mid~}
\newcommand{\sstmt}[2]{{#1}\ \kw{:=}\ {#2}}
\newcommand{\ssep}[0]{\kw{;}\ }

\newcommand{\sfun}[3]{\sstmt{#1}{\efun{#2}{#3}}}
\newcommand{\sapp}[3]{\sstmt{#1}{\eapp{#2}{#3}}}

\newcommand{\stuple}[2]{\sstmt{#1}{\etuple{#2}}}
\newcommand{\sproj}[3]{\sstmt{#1}{\eproj{#2}{#3}}}
\newcommand{\sprim}[2]{\sstmt{#1}{\eprim{#2}}}
\newcommand{\stask}[3]{\sstmt{#1}{\etask{#2}{#3}}}
\newcommand{\sret}[1]{\kw{ret}\ {#1}}

\newcommand{\obsprim}[3]{{#2} \Uparrow_{#1} {#3}}
\newcommand{\freevars}[1]{\operatorname{fv}\left(#1\right)}
\newcommand{\subst}[3]{{#1}\left\{{#2}/{#3}\right\}}
\newcommand{\freshen}[2]{\operatorname{freshen}_{#1}\left(#2\right)}
\newcommand{\stepreduce}[3]{{#1} \Rightarrow^S_{#2} {#3}}
\newcommand{\stepdispatch}[3]{{#1} \Rightarrow^D_{#2} {#3}}
\newcommand{\stepresolve}[3]{{#1} \Rightarrow^R_{#2} {#3}}
\newcommand{\stepfull}[3]{{#1} \Rightarrow_{#2} {#3}}
\newcommand{\stepat}[4]{{#1} \Rightarrow_{#2}^{#3} {#4}}

\newcommand{\vsubst}[2]{\left\{{#1}/{#2}\right\}}

\newcommand{\decomp}[1]{\left\lbrack{#1}\right\rbrack}
\newcommand{\expand}[1]{\left\lbrack\hspace{-0.25em}\left\lbrack{#1}\right\rbrack\hspace{-0.25em}\right\rbrack}

\newcommand{\initstate}[2]{\operatorname{initial}_{#1}({#2})}
\newcommand{\erasestate}[1]{\operatorname{erase}({#1})}
\newcommand{\ctxident}[1]{\operatorname{ident}({#1})}

\newcommand{\sreplace}[4]{\operatorname{replace}({#1}, {#2}, {#3}, {#4})}

\title{Opportunistically Parallel Lambda Calculus}

\author{Stephen Mell}
\orcid{0009-0003-7469-8974}
\affiliation{%
  \institution{University of Pennsylvania}
  \country{USA}
}
\email{sm1@cis.upenn.edu}

\author{Konstantinos Kallas}
\orcid{0000-0002-8984-6648}
\affiliation{%
  \institution{University of California, Los Angeles}
  \country{USA}
}
\email{kkallas@ucla.edu}

\author{Steve Zdancewic}
\orcid{0000-0002-3516-1512}
\affiliation{%
  \institution{University of Pennsylvania}
  \country{USA}
}
\email{stevez@cis.upenn.edu}

\author{Osbert Bastani}
\orcid{0000-0001-9990-7566}
\affiliation{%
  \institution{University of Pennsylvania}
  \country{USA}
}
\email{obastani@seas.upenn.edu}

\date{April 2024}
\begin{abstract}
Scripting languages are widely used to compose \emph{external calls} such as native libraries and network services. In such scripts, execution time is often dominated by waiting for these external calls, rendering traditional single-language optimizations ineffective. To address this, we propose a novel \emph{opportunistic} evaluation strategy for scripting languages based on a core lambda calculus that \emph{automatically} dispatches independent external calls in parallel and streams their results. We prove that our approach is confluent, ensuring that it preserves the programmer's original intent, and that it eventually executes every external call. We implement this approach in a scripting language called \toolname{}. We demonstrate the versatility and performance of \toolname{}, focusing on programs that invoke heavy external computation through the use of large language models (LLMs) and other APIs.  Across five scripts, we compare to several state-of-the-art baselines and show that opportunistic evaluation improves total running time (up to $6.2\times$) and latency (up to $12.7\times$) compared to standard sequential Python, while performing very close (between $1.3\%$ and $18.5\%$ running time overhead) to hand-tuned manually optimized asynchronous Rust. For Tree-of-Thoughts, a prominent LLM reasoning approach, we achieve a $6.2\times$ performance improvement over the authors' own implementation.
\end{abstract}

\setcopyright{cc}
\setcctype{by}
\acmDOI{10.1145/3763143}
\acmYear{2025}
\acmJournal{PACMPL}
\acmVolume{9}
\acmNumber{OOPSLA2}
\acmArticle{365}
\acmMonth{10}
\received{2025-03-26}
\received[accepted]{2025-08-12}



\begin{document}
\maketitle
\section{Introduction}

A key application of scripting languages such as Python and the shell 
is to ``glue'' together \emph{external calls}, i.e., algorithms implemented as foreign calls to low-level languages or network calls to remote APIs, whose implementations are opaque to the interpreter~\cite{ousterhout1998scripting}. Unlike other programs, performance bottlenecks in glue scripts are typically in the external calls, rather than in script evaluation itself. 
While each individual external call is usually heavily optimized by its developers, their composition is not---in glue scripts there are often opportunities to execute independent external calls in parallel and stream results between them.

Traditionally, it is up to the script developer to exploit parallelism or streaming across calls, but doing so may result in complex multithreaded code, undermining the simplicity and usability of scripting languages.
At the same time, automatically exposing parallelization and streaming opportunities across external calls is challenging because they are often interwoven with complex control flow. Traditional parallelizing compilers~\cite{apostolakis2020perspective,johnson2012speculative} and evaluation strategies~\cite{blelloch95parallelism,blelloch99scheduling} focus on identifying dependencies in a single language without external calls. Recent work on parallelizing shell glue scripts with external calls~\cite{pash,pash-osdi-2022} can only expose limited parallelism across control-flow, focusing instead on parallelizing contiguous pipelines.

We propose a novel, general-purpose, higher-order scripting language designed to automatically execute external calls in a parallel and streaming way, while remaining consistent with the sequential evaluation semantics programmers expect.

A key assumption that we make is that external call dependencies can be fully determined by the call arguments, which holds for many categories of external calls, e.g., API calls and command utilities.
This assumption allows the language to evaluate two calls out of order and in parallel unless they are connected by an explicit data dependency. 
Our approach has three components: (1) a novel core calculus, \calcname{} that has external calls as first-class citizens, (2) a novel \emph{opportunistic} evaluation strategy that parallelizes independent external calls during the execution of the program,
and (3) extensions supporting partial data and control-flow, providing automatic streaming for real-world scripts.

Our evaluation strategy achieves significantly more parallelism than existing evaluation strategies. On one hand, strict evaluation executes all external calls sequentially, and on the other hand, lazy evaluation can automatically expose parallelism (e.g., Parallel Haskell~\cite{marlow2013parallel}) but might duplicate or drop effects.\footnote{As Simon Peyton Jones wrote of Haskell, ``[the] bottom line is that laziness and side effects are, from a practical point of view, incompatible.''~\cite{jones2001tackling}} While monads address this issue, they serialize \emph{all} effects, including random sampling,
diminishing the opportunities for parallelization.
In contrast, opportunistic evaluation guarantees that external calls are not dropped or duplicated. Our opportunistic semantics is presented as a set of reduction rules: when an external call is dispatched, it replaces the call with a hole, which is filled when the call completes, while \emph{internal} function calls (i.e., traditional calls to functions implemented in the calculus) are stepped in parallel. External and internal reduction steps are interleaved nondeterministically.
We prove that  this evaluation strategy is confluent, ensuring that opportunistic evaluation for terminating programs makes the same external calls as sequential, call-by-value evaluation, preserving the programmer's original intent.

On top of \calcname{}, we build \toolacronym{}, a scripting language with support for standard control flow constructs and automatic streaming. The key challenge with streaming is that our opportunistic semantics only fills a hole once an external call has completely finished running. We address this by allowing external calls to return \calcname{} terms with holes representing the values that remain to be computed. For instance, consider an external call that streams a list. As soon as the next element of the list has been computed, it can return a pair consisting of this element together with a hole to be filled as additional list elements become available. Internal computation on this list may be able to proceed opportunistically based just on the partially returned prefix. To make such interactions work smoothly, we use \textit{Church encodings} to represent \calcname{} datatypes, including lists; as we explain, this approach provides a natural notion of \emph{partial} value suited for streaming.

\add{We have implemented \toolname{} as a} Python framework that allows users to write \calcname{} programs using a subset of Python's syntax and make external calls to arbitrary Python functions, such as popular machine learning libraries. The programs are evaluated using our opportunistic semantics, providing automatic parallelization without any additional user effort.

To evaluate the efficacy of our approach, we focus on glue scripts that compose calls to remote APIs, including large language models (LLMs). LLM glue scripts are a particularly good domain for several reasons:
\begin{itemize}
\item \emph{High capacity for parallelization:} State-of-the-art LLMs are typically provided as remote services; providers handle the scheduling challenges of parallel execution, and their scale is so large that a single user effectively cannot exceed the provider's capacity for parallelization.
\item \emph{Long execution times:} LLMs can take 5-10 seconds to return all of their output, which makes the overhead of the scripting language negligible in comparison.
\item \emph{Streaming:} LLM calls often stream their outputs to reduce latency---e.g., they can produce their initial result in just a few milliseconds despite taking several seconds to complete.
\end{itemize}
However, despite numerous languages and frameworks proposed to help developers write these kinds of scripts~\cite{langchain,guidance,sglang,lmql,dspy}, none of them focus on automatic parallelization and streaming.

We evaluate \toolname{} on five such scripts, 
including scripts for constrained decoding~\cite{poesia2022synchromesh}, retrieval-augmented generation~\cite{schick2024toolformer}, tool calling~\cite{schick2024toolformer}, and tree-of-thoughts~\cite{tree-of-thoughts}. We evaluate on latency and total running time against both standard Python and bash implementations as well as state-of-the-art parallelization and streaming baselines (PaSh~\cite{pash}, Apache Flink~\cite{flink} and SGLang~\cite{sglang}). 
\toolname{} achieves significant speedups, improving latency by up to $12.7\times$ and total time by up to $6.2\times$ compared to sequential Python, while achieving comparable performance to manually optimized asynchronous Rust (up to $4.8\%$ latency and $18.5\%$ running time overhead).
For Tree-of-Thoughts~\cite{tree-of-thoughts} a prominent LLM reasoning approach, \toolname{} achieves a $6.2\times$ performance improvement over the authors' own Python implementation~\cite{tot-impl}.

In summary, our contributions are:
\begin{itemize}
\item The design of a core calculus \calcname{}, with first-class support for nondeterministic external calls
(Sections~\ref{sec:corecalculus:syntax} and~\ref{sec:corecalculus:semantics}).
\item A non-strict, non-lazy \emph{opportunistic} evaluation strategy for \calcname{}, achieving a high degree of parallelization \emph{automatically} while executing every possible external call and matching the behavior of strict evaluation for terminating programs (Section~\ref{sec:corecalculus:eval}).
\item A scripting language \toolname{}, built on top of \calcname{}, offering structured control flow and automatic parallelism and streaming (Section~\ref{sec:epiclang}).
\item An implementation of \toolname{} where programs can be written using a subset of Python syntax, and where arbitrary Python functions, including existing machine learning libraries, can be used as external calls (Section~\ref{sec:impl}).
\item An evaluation of \toolname{} on several real-world LLM scripting applications, comparing against both state-of-the-art parallelization and streaming systems and manually optimized asynchronous Rust (Section~\ref{sec:evaluation}).
\end{itemize}

\section{Motivating Example}
\label{sec:motivating}

We begin with a simple example to demonstrate the kinds
of LLM glue scripts that are often written by practitioners, usually in Python. 
We then show how the example contains a large amount of latent parallelism, and finally how our approach is able to automatically exploit that parallelism.

A common pattern in LLM programming is to make an LLM call, parse the output (e.g. into a list), and then make subsequent LLM calls based on the result. Adding structure in this way, rather than asking the LLM to solve a problem end-to-end, can increase both the interpretability and reliability of the results~\cite{tree-of-thoughts,schick2024toolformer}.
For example, suppose we have two external calls, implemented using LLMs: \kw{get_cities_in}, which takes a location name (as a string) and returns a list of city names in some order (e.g. personalized relevance); and \kw{get_excursions_in}, which takes a city name and returns text about excursions in that location. If we'd like to find tourist destinations in Oceania, we can first ask the LLM for locations in Oceania, and then ask for some fun excursions in each location:

\begin{inlinebox}
    \footnotesize
    \begin{minipage}[c]{0.36\textwidth}
\begin{minted}[escapeinside=||]{python}
for city in get_cities_in("Oceania"):
    print(&stdout, city)
    excursions = get_excursions_in(city)
    print(&stdout, excursions)
\end{minted}
    \end{minipage}
    \hfill
    $\Rightarrow$
    \hfill
    \begin{minipage}[c]{0.5\textwidth}
\begin{minted}[escapeinside=||]{python}
# print(&stdout, "Honolulu")
# print(&stdout, "Located in the US state of...")
# print(&stdout, "Jakarta")
# print(&stdout, "In Jakarta, you can find...")
# ...
\end{minted}
    \end{minipage}
\end{inlinebox}
Unfortunately, this program takes a long time to execute: each API call takes about 4 seconds to run, and the whole program takes about 51 seconds. However, the program exhibits a high degree of parallelism: the excursions for each city could be fetched in parallel. Further, since the LLM responses stream, we could begin processing cities before all of them are available. Taking advantage of these opportunities brings the total running time down to 9.2 seconds, a $5.5\times$ speedup.

Recent work has introduced several DSLs for different aspects of LLM programming, including LangChain~\cite{langchain}, LlamaIndex~\cite{llamaindex}, Guidance~\cite{guidance}, SGLang~\cite{sglang}, and DSPy~\cite{dspy}. However, none of them automatically parallelize this kind of loop, and some cannot handle control flow at all.

Many languages and parallelism frameworks support some kind of ``parallel map'' construct. In this case, we cannot just use parallel map, since that might re-order our \kw{print} statements, but we could refactor the code into a parallel map followed by a sequential loop:
\begin{inlinebox}
    \begin{minipage}[c]{0.53\textwidth}
\begin{minted}{python}
cities = get_cities_in("Oceania")
excursions = parallel_map(cities, get_excursions_in)
for city, excursion in zip(cities, excursions):
    print(&stdout, city)
    print(&stdout, excursion)
\end{minted}
    \end{minipage}
\end{inlinebox}

However, this requires programmers to make code changes and reason explicitly about the parallelism in their programs. Particularly with scripting languages, developers tend not to be expert programmers, and such changes may be prohibitive. For example, in the Tree-of-Thoughts~\cite{tree-of-thoughts} implementation~\cite{tot-impl}, no parallelization is used of any kind, despite there being several instances where parallel map could be used to improve performance by many times.

Further, the use of parallel map degrades latency (the time until the first program output) versus the original sequential example: parallel map returns when all results have been fetched, whereas the sequential example can at least print the first city's excursions as soon as they are available. The importance of latency to end-user experience is why LLM APIs provide streaming results. While other parallelization options exist with lower latency, e.g. manually creating processes and passing messages between them, writing such code is complex and error-prone.

Instead, we propose using a language that is parallel \emph{by default}. Practitioners write their high-level intent, ``as if'' execution were sequential; then, the runtime automatically executes the program as eagerly and with as much parallelism as possible. The semantics of our language is a rewrite-style operational semantics, and has three kinds of steps: \emph{reduction} steps, which are ordinary, internal computation (e.g. ordinary lambda calculus reduction); \emph{dispatch} steps, where an external call is initiated (e.g. by sending a network request) and the call site is a replaced by a ``task placeholder''; and \emph{resolution} steps, when the external call finishes (e.g. by receiving a network response) and the task placeholder is replaced by the result. Crucially, programs will continue executing, despite the presence of task placeholders.

\begin{figure}[t]
    \vskip 0pt
    \rule{\textwidth}{1pt}
    \hfill
      \begin{subfigure}[t]{0.33\textwidth}
          \captionsetup{singlelinecheck=false}
          \centering
          \footnotesize
          \begin{minted}{python}
for city in get_cs("Oc"):
    print(&stdout, city)
    exc = get_exc(city)
    print(&stdout, exc)
          \end{minted}
          \caption{Initial}
      \end{subfigure}\hfill
      \begin{subfigure}[t]{0.33\textwidth}
          \centering
          \captionsetup{singlelinecheck=false}
          \footnotesize
          \begin{minted}{python}
for city in TASK[A0]:
    print(&stdout, city)
    exc = get_exc(city)
    print(&stdout, exc)
          \end{minted}
          \caption{Dispatch: $\kw{get_cs} \Rightarrow \kw{A0}$}
      \end{subfigure}\hfill
      \begin{subfigure}[t]{0.33\textwidth}
          \centering
          \captionsetup{singlelinecheck=false}
          \footnotesize
          \begin{minted}{python}
for city in ["Ho"] + TASK[A1]:
    print(&stdout, city)
    exc = get_exc(city)
    print(&stdout, exc)
          \end{minted}
          \caption{Resolve: $\kw{A0} \Rightarrow (\kw{["Ho"] + A1})$}
      \end{subfigure}
      \\ \vspace{0.5em}\rule{\textwidth}{0.5pt}
      \begin{subfigure}[t]{0.33\textwidth}
          \centering
          \captionsetup{singlelinecheck=false}
          \footnotesize
          \begin{minted}{python}
print(&stdout, "Ho")
exc = get_exc("Ho")
print(&stdout, exc)
for city in TASK[A1]:
    print(&stdout, city)
    exc = get_exc(city)
    print(&stdout, exc)
          \end{minted}
          \caption{Reduce: unroll loop}
      \end{subfigure}\hfill
      \begin{subfigure}[t]{0.33\textwidth}
          \centering
          \captionsetup{singlelinecheck=false}
          \footnotesize
          \begin{minted}{python}
# print(&stdout, "Ho")
exc = get_exc("Ho")
print(&stdout, exc)
for city in TASK[A1]:
    print(&stdout, city)
    exc = get_exc(city)
    print(&stdout, exc)
          \end{minted}
          \caption{Dispatch: \kw{print}}
      \end{subfigure}\hfill
      \begin{subfigure}[t]{0.33\textwidth}
          \centering
          \captionsetup{singlelinecheck=false}
          \footnotesize
          \begin{minted}{python}
# print(&stdout, "Ho")
exc = TASK[B0]
print(&stdout, exc)
for city in TASK[A1]:
    print(&stdout, city)
    exc = get_exc(city)
    print(&stdout, exc)
          \end{minted}
          \caption{Dispatch: $\kw{get_exc} \Rightarrow \kw{B0}$}
      \end{subfigure}
      \\ \vspace{0.5em}\rule{\textwidth}{0.5pt}
      \begin{subfigure}[t]{0.33\textwidth}
          \centering
          \captionsetup{singlelinecheck=false}
          \footnotesize
          \begin{minted}{python}
# print(&stdout, "Ho")
exc = TASK[B0]
print(&stdout, exc)
for city in ["Ja"]:
    print(&stdout, city)
    exc = get_exc(city)
    print(&stdout, exc)
          \end{minted}
          \caption{Resolve: $\kw{A1} \Rightarrow \kw{["Ja"]}$}
      \end{subfigure}\hfill
      \begin{subfigure}[t]{0.33\textwidth}
          \centering
          \captionsetup{singlelinecheck=false}
          \footnotesize
          \begin{minted}{python}
# print(&stdout, "Ho")
exc = TASK[B0]
print(&stdout, exc)

print(&stdout, "Ja")
exc = get_exc("Ja")
print(&stdout, exc)
          \end{minted}
          \caption{Reduce: unroll loop}
      \end{subfigure}\hfill
      \begin{subfigure}[t]{0.33\textwidth}
          \centering
          \captionsetup{singlelinecheck=false}
          \footnotesize
          \begin{minted}{python}
# print(&stdout, "Ho")
exc = TASK[B0]
print(&stdout, exc)

print(&stdout, "Ja")
exc = TASK[C0]
print(&stdout, exc)
          \end{minted}
          \caption{Dispatch: $\kw{get_exc} \Rightarrow \kw{C0}$}
      \end{subfigure}\hfill
      \\ \vspace{0.5em}\rule{\textwidth}{0.5pt}
      \begin{subfigure}[t]{0.33\textwidth}
          \centering
          \captionsetup{singlelinecheck=false}
          \footnotesize
          \begin{minted}{python}
# print(&stdout, "Ho")
exc = TASK[B0]
print(&stdout, exc)

print(&stdout, "Ja")

print(&stdout, "...")
          \end{minted}
          \caption{Resolve: $\kw{C0} \Rightarrow \kw{"..."}$}
      \end{subfigure}\hfill
      \begin{subfigure}[t]{0.33\textwidth}
          \centering
          \captionsetup{singlelinecheck=false}
          \footnotesize
          \begin{minted}{python}
# print(&stdout, "Ho")

print(&stdout, "...")

print(&stdout, "Ja")

print(&stdout, "...")
          \end{minted}
          \caption{Resolve: $\kw{B0} \Rightarrow \kw{"..."}$}
      \end{subfigure}\hfill
      \begin{subfigure}[t]{0.33\textwidth}
          \centering
          \captionsetup{singlelinecheck=false}
          \footnotesize
          \begin{minted}{python}
# print(&stdout, "Ho")

# print(&stdout, "...")

# print(&stdout, "Ja")

# print(&stdout, "...")
          \end{minted}
          \caption{Dispatch: \kw{print} ($\times 3$)}
      \end{subfigure}\hfill
      \\ \vspace{0.5em}\rule{\textwidth}{1pt}
      \caption{A conceptual demonstration of how our approach achieves automatic parallelization and streaming for the motivating example program. Our approach is based on a rewrite-syle operational semantics, involving three kinds of steps: \emph{reduction} steps are ordinary rewrites, internal to the term, as in lambda calculus; \emph{dispatch} steps occur when an external call is made, replacing the call with a placeholder \kw{TASK} (function $\Rightarrow$ task); \emph{resolution} steps occur when an external call returns and the TASK is replaced by an expression (task $\Rightarrow$ expression). The captions indicate what step led to that term. The initial term is shown in (a) and the final term in (l). Prints that have been performed are shown with a comment (\#).}
      \label{fig:motivating-example}
  \end{figure}

As an example of opportunistic evaluation, consider in Figure~\ref{fig:motivating-example} our example from earlier. (For reasons of space, strings and function names have been shortened, and the LLM only returns two cities.) The initial term is shown in (a). First, because the argument to \kw{get_cs} is a value, the external call is dispatched, sending the network request to the LLM provider, and the call syntax is replaced with a \kw{TASK[A0]} placeholder, to get (b). The computation is now blocked until some result is received. When a city name has been streamed from the LLM, the \kw{TASK[A0]} resolves to \kw{["Ho"] + TASK[A1]}, to get (c). Now, since the first list element is known, the for loop can be partially unrolled via reduction, to get (d). Now, since the first \kw{print} statement has a concrete argument, the runtime can dispatch it, printing to the user, to get (e). No \kw{TASK} is produced, since \kw{print} does not return anything. Since the argument to the first \kw{get_exc} is known, it is dispatched, sending another network request to the LLM provider and replacing the call with \kw{TASK[B0]}, to get (f). The computation is again blocked until some result is received, either for task \kw{A1} or \kw{B0}. Suppose \kw{A1} returns first, resolving to \kw{["Ja"]}, to get (g). Now, since the entire list being looped over is known, the loop can be unrolled entirely via reduction, to get (h). Now, even though \kw{"Ja"} is a known argument to \kw{print}, it cannot be dispatched yet, as it must wait for the preceding \kw{print} call to execute. The call to \kw{get_exc}, however, can be dispatched and replaced with \kw{TASK[C0]}, to get (i). The computation is again blocked until either \kw{B0} or \kw{C0} resolves. Suppose \kw{C0} resolves first, to \kw{"..."}, to get (j). Now there are two \kw{print} calls that have arguments, but are blocked by the preceding \kw{print} call. Finally, suppose \kw{B0} resolves to \kw{"..."}, to get (k). Now all three \kw{print} calls can be dispatched in rapid succession, to get (l).

In this example, we were able to both dispatch a \kw{get_exc} call while \kw{get_cs} was only partially completed, and we were able to have both \kw{get_exc} calls run in parallel, in contrast to standard sequential execution. Because LLM scripts have many time-consuming but parallelizable external calls, this sort of streaming and parallel dispatch of external calls can dramatically improve the runtime. On the motivating example (where \kw{get_cs} actually returns about ten cities, rather than the two in the example) sequential evaluation has a total runtime of 51 seconds and a latency (time until first print) of 7.4 seconds. Opportunistic evaluation brings these down to 9.2 and 1.3 seconds, for $5.5\times$ and $5.6\times$ speedups,  respectively.

\section{\calcname{} Calculus}
\label{sec:corecalculus}

At the center of our approach is a core calculus, \calcname{}, which enables out-of-order, ``opportunistic'' evaluation of external calls. This section presents \calcname{}, starting from a high-level overview (Section~\ref{sec:corecalculus:overview}), before formally presenting the syntax (Section~\ref{sec:corecalculus:syntax}), semantics (Section~\ref{sec:corecalculus:semantics}), and evaluation strategy (Section~\ref{sec:corecalculus:eval}).

\subsection{Overview}
\label{sec:corecalculus:overview}

\calcname{} is based on lambda calculus with tuples, but differs from traditional calculi in three ways:
(1) an explicit treatment, in both the syntax and semantics, of external calls; 
(2) the syntactic use of let-bindings rather than nested expressions, similar to ANF; 
and (3) a novel, non-strict, non-lazy ``opportunistic'' evaluation strategy.

\paragraph{External Calls}
One of the key design choices is the explicit treatment of external calls. We assume the existence of a set of external values, which we call \emph{primitives} (in our implementation, these are just Python values). Syntactically, there are two additions beyond lambda calculus: primitives and external \emph{tasks}. An \emph{external call} is a function call where the function is a primitive, rather than an (internal) function definition. (For simplicity, we also assume that the arguments to an external call are primitive.) Semantically, there are two additions. First, external calls step to tasks via the \emph{dispatch} semantics. Dispatch is when the runtime begins executing the external operation, e.g. by sending a network request. The task syntax serves as a placeholder for the result of the running operation. Second, when the operation has finished, the task placeholder is replaced by the result via the \emph{resolution} semantics. The dispatch and resolution semantics are defined in Figure~\ref{fig:semantics:rules} and detailed in Section~\ref{sec:corecalculus:semantics}.
Because external calls are, by definition, opaque to the interpreter, the semantics assumes the existence of an \emph{external semantics} relation that maps an external function and external arguments to a set of results.

We show that \calcname{} is confluent, even in the presence of nondeterministic external calls. 
In particular, for any program, if the result of each external call (e.g. coin flip) is fixed between runs, then the order of evaluation does not affect the final result. Thus, while a program may have a nondeterministic result, the nondeterminism is due only to external calls, not due to internal evaluation order.
To guarantee confluence, we assume that the behavior of an external call depends only on its arguments, and thus that external calls with no data dependency may be executed in arbitrary order. See Section~\ref{sec:epiclang:sequencing} for a discussion of enforcing sequencing.

\paragraph{Statement Identifiers and ANF}
Proving confluence relies on the ability to identify a specific external call and what it resolved to. Often, external calls will not exist in the initial term, but rather will be dynamically generated as evaluation proceeds. Thus we must have a unique way of identifying each external call that might be generated.

To accomplish this, the syntax of our language is in ANF~\cite{anf} style---rather than having nested expressions, expressions are a sequence of let-bindings, followed by a return variable. Each statement binds a variable to some operation: defining a function, calling a function, constructing a tuple, or projecting out of a tuple (accommodating lambda calculus with tuples); or to an external primitive or task (accommodating external calls). See Section~\ref{sec:corecalculus:syntax} for details and Figure~\ref{fig:syntax:example}~(a) for the grammar.

Thus, each statement in an initial expression can be identified with a number, indicating its position in the list of statements. 
Execution proceeds by replacing a statement (e.g. a function call) with a list of statements (e.g. the body of the function being called). Such dynamically-generated statements can be identified by a \emph{list} of numbers. For example, for a function call with identifier $(3)$, the $n$ statements generated by stepping it would have identifiers $(3,1), (3,2), \ldots, (3,n)$. Importantly, these statement identifiers are consistent regardless of evaluation order. Moreover, they are needed only for \emph{proving} confluence; such metadata need not be explicitly maintained by an implementation. See Section~\ref{sec:corecalculus:semantics} for the definition and discussion of statement identifiers.

\paragraph{Opportunistic Evaluation}
Since confluence ensures that it is sound to step \calcname{} terms in any order, we need an evaluation strategy that takes advantage of this to produce parallel execution.
Strict (call-by-value) evaluation would unnecessarily block on external calls or diverging computation, while lazy (call-by-need) evaluation might drop an external call (i.e. never invoke a particular external call statement). Instead, our evaluation strategy, which we call \emph{opportunistic}, simultaneously
takes all steps that are available to the semantics. For terminating programs, it makes exactly the same external calls as strict evaluation. However, like lazy evaluation, it is not unnecessarily blocked by external tasks or diverging computations. We show that opportunistic evaluation dispatches every external call that is reachable via the semantics, which is strictly more than either strict or lazy evaluation. See Section~\ref{sec:corecalculus:eval} for details and Figure~\ref{fig:semantics:rules} for the semantics rules.

\subsection{Syntax}
\label{sec:corecalculus:syntax}

\paragraph{Grammar} The syntax, shown in Figure~\ref{fig:syntax:example}~(a),
consists of expressions ($e$), statements ($s$), and operations ($o$). The grammar is parameterized by a set of primitives ($c$), which in our implementation are Python objects. A primitive expression ($\bar{c}$) is either a primitive or a tuple of primitive expressions. Variables ($x, f$) are identifiers. An expression is either a statement followed by an expression $(s \ssep e)$ or a return variable $(\sret{x})$. We elide \kw{ret} when unambiguous and often use line breaks instead of semicolons for readability. Note that an expression is equivalent to a list of statements $\vec{s}$ followed by a return variable---we will sometimes write expressions this way, as $(\vec{s} \ssep{} x)$. The empty list of statements is written $\varnothing$. Statements bind variables to operations $(\sstmt{x}{o})$. An operation is either a function definition $(\efun{x}{e})$, with parameter $x$ and body $e$, a call $(\eapp{f}{x})$ of function $f$ with argument $x$, a tuple $(\etuple{x_1, \ldots, x_n})$ with components $x_i$, a projection $(\eproj{i}{x})$ of component $i$ from a tuple $x$, a primitive $(\eprim{c})$, or a task $(\etask{c_f}{\bar{c}_x})$ running primitive function $c_f$ with argument $\bar{c}_{x}$. We say that a statement $s$ is \emph{top-level} in an expression $e$ if it is not inside of a function definition. See Figure~\ref{fig:syntax:example}~(b) for an example term.

\begin{figure}
    \begin{subfigure}{0.95\textwidth}
        \centering
        \begin{align*}
            e \Coloneqq\ s \ssep{} e \sep{} \sret{x} \qquad\qquad
            s \Coloneqq\ \sstmt{x}{o} \qquad\qquad
            \bar{c} \Coloneqq\ c \sep{} (\bar{c}_1, \ldots, \bar{c}_n)
        \end{align*}
        \begin{align*}
            o \Coloneqq&\ x
            \sep{} \efun{x}{e}
            \sep{} \eapp{f}{x}
            \sep{} \etuple{x_1, \ldots, x_n}
            \sep{} \eproj{i}{x}
            \sep{} \eprim{c}
            \sep{} \etask{c_f}{\bar{c}_x}
        \end{align*}
        \caption{}
    \end{subfigure} \\
    \begin{subfigure}{0.45\textwidth}
        \centering
        \footnotesize
        \begin{verbatim}
x0 := <"Oceania">
cities := get_cities_in x0
for_body := fun x1:
    stdout1 := prj 1 x1
    city := prj 2 x1
    excursion := get_excursions_in city
    x2 := (stdout1, excursion)
    stdout2 := print x2
    ret stdout2
x3 := (cities, stdout, for_body)
stdout3 := fold x3
x4 := ()
ret x4
        \end{verbatim}
        \caption{}
    \end{subfigure}\hfill
    \begin{subfigure}{0.45\textwidth}
        \centering
        \footnotesize
        \begin{verbatim}

cities := get_cities_in <"Oceania">
for_body := fun (stdout, city):


    excursion := get_excursions_in (city)

    stdout := print (stdout, excursion)
    stdout

stdout := fold (cities, stdout, for_body)

()
        \end{verbatim}
        \caption{}
    \end{subfigure}
    \caption{The grammar of \calcname{} (a). The motivating example from Section~\ref{sec:motivating}, in the formal grammar of \calcname{} (b) and in sugared form (c). The example is well-formed in contexts containing \kw{cities_in}, \kw{excursion_in}, \kw{fold}, \kw{print}, and \kw{stdout}.}
    \label{fig:syntax:example}
\end{figure}

\paragraph{Well-Formedness}
An expression $e$ is \emph{well-formed in a context $\Gamma$} (written $\Gamma \vdash e$) in the standard sense: variables that are used in an expression must be in the context, with binding occurrences $\efun{x}{e}$ and $\sstmt{x}{o} \ssep{} e$ adding $x$ to the context. For simplicity in the semantics, well-formedness also prohibits re-binding of variable names. See Appendix~\ref{appendix:wf}
for the formal definition. The example in Figure~\ref{fig:syntax:example} is well-formed in contexts containing \kw{get_cities_in}, \kw{get_excursions_in}, \kw{fold}, \kwtt{print}, and \kw{stdout}. We say \emph{$e$ is well-formed} if it is well-formed in some $\Gamma$.

\paragraph{Syntactic Sugar}
Informally we adopt several conventions to aid in readability. As mentioned, we elide \kw{ret} when unambiguous and may use line breaks instead of semicolons for readability. Variables that are re-bound (e.g. $\sstmt{x}{o} \ssep{} \sstmt{x}{o'}$), desugar to unique variable names (e.g. $\sstmt{x}{o} \ssep{} \sstmt{x'}{o'}$). Tuples in variable binding sites (e.g. $\sstmt{\etuple{y, z}}{x}$) desugar to projections (e.g. $\sstmt{y}{\eproj{1}{x}} \kw{;}$ $\sstmt{z}{\eproj{2}{x}}$). Operations in variable positions (e.g. $\stuple{z}{\etuple{w,x}, y}$) desugar to definitions of intermediate variables (e.g. $\stuple{u}{w,x} \ssep{} \stuple{z}{u, y}$). See Figure~\ref{fig:syntax:example}~(b) and~(c) for an example of desugaring.

\paragraph{Syntactic Operations}
We write $\freevars{e}$ to mean the free variables of $e$, i.e. those that occur but are not bound. We write $\subst{e}{x}{y}$ to mean the result of replacing all occurrences of a free variable $y$ in $e$ with $x$. We write $\freshen{E}{e}$ to mean renaming the top-level bound variables of $e$ such that they do not occur in $e$ or an evaluation context $E$ (defined in Section~\ref{sec:corecalculus:semantics}). We write $\freshen{E}{\efun{x}{e}}$ to mean that we additionally rename $x$. We write $s \in e$ to mean that $s \in \vec{s}$ for $e = \mquot{\vec{s} \ssep{} x}$.

\subsection{Semantics}
\label{sec:corecalculus:semantics}
We give \calcname{} a rewrite-style operational semantics, which consists of three parts: \emph{reduction} steps, \emph{dispatch} steps, and \emph{resolution} steps. Reduction steps capture the standard semantics of lambda calculus with tuples, adapted for the ANF style of the syntax: stepping function definition/call pairs and tuple construction/projection pairs. Dispatch and resolution are novel, handling the initiation and termination of external calls.

The semantics is essentially a relation $\stepfull{e}{}{e'}$, stepping an expression $e$ to $e'$. However, for the purposes of proving properties, we slightly enrich expressions $e$ to ``labeled expressions'' $p$, and the semantics steps $\stepfull{p}{}{p'}$.
Each rule in the semantics matches a single top-level statement in $p$ and replaces it with zero or more new statements to produce $p'$.

\paragraph{Statement Labels and Labeled Expressions}
For the purposes of stating and proving certain properties, it is useful to be able to uniquely identify statements in an expression, including those generated during execution.
Thus, the semantics operates on \emph{labeled expressions} $p$, which are expressions, but where the top-level statements are tagged with \emph{statement labels} $\ell$, nonempty sequences of natural numbers (see Figure~\ref{fig:semantics:ctxgrammar}).

We can erase the labels to get an expression: $\erasestate{\ell_1 : s_1 \ssep{} \ldots \ssep{} \ell_n : s_n \ssep x} = \mquot{s_1 \ssep{} \ldots \ssep{} s_n \ssep x}$. We implicitly erase labels when using a $p$ as an $e$. We can also initialize an expression $e$ with sequential, singleton labels:
$\initstate{}{s_1 \ssep{} \ldots \ssep{} s_n \ssep{} x} = (1 : s_1 \ssep{} \ldots \ssep{} n : s_n \ssep{} x)$.

In order for these labels to be useful, they must be unique in an expression. However, uniqueness is not quite strong enough of a property to be stable under the semantics: we will see later that if we had a labeled expression $(1: s_1 \ssep{} 1,1: s_2)$, stepping $s_1$ might produce a new statement $s_3$ with label $(1,1)$, resulting in $(1,1: s_3 \ssep{} 1,1: s_2)$ and causing $(1,1)$ not to be a unique label. Thus we introduce a stronger notion of \emph{label-independence}, which requires that, in addition to being unique, no label is a prefix of another label. Thus $(1: s_1 \ssep{} 1,1: s_2)$ would be prohibited, but $(1: s_1 \ssep{} 2: s_2)$ or $(1, 1: s_1 \ssep{} 1,2: s_2)$ would not. A labeled expression is \emph{well-formed} if it is label-independent and it is well-formed as an (unlabeled) expression. Unless otherwise specified, we will assume that labeled expressions are well-formed.

\begin{figure}
    \begin{align*}
        \ell \Coloneqq i \sep{} \ell, i \qquad\qquad
        p \Coloneqq \ell : s \ssep{} p \sep{} x \qquad\qquad
        E \Coloneqq \ell : s \ssep{} E \sep{} \ell : [] \ssep{} p
    \end{align*}
    \caption{The grammars for statement labels $\ell$, labeled expressions $p$, and evaluation contexts $E$. $i$ ranges over natural numbers.}
    \label{fig:semantics:ctxgrammar}
\end{figure}

\paragraph{Evaluation Contexts}
An \emph{evaluation context} $E$ is a labeled expression, but with one top-level statement replaced by a \emph{hole}, ``$[]$'' (defined in Figure~\ref{fig:semantics:ctxgrammar}). We write $p = E\decomp{s}$ to mean that replacing the hole in $E$ with $s$ yields $p$. Crucially, we can also replace the hole in $E$ with a list of statements $\vec{s}$, written $E\expand{\vec{s}}$, similar to a ``flat map'' operation. If the hole has statement label $\ell$, then the $i$-th new statement has label $(\ell, i)$---i.e. label $\ell$ is replaced by $n$ statements with labels $(\ell, 1)$ through $(\ell, n)$. See Figure~\ref{fig:semantics:ctxsubst} for the formal definitions of these replacement operations. Let $\ctxident{E}$ give the statement label of the hole in $E$ (i.e. $\ctxident{\ell : [] \ssep{} p} = \ell$). We define $\freevars{E'}$ and $s \in E$ based on their versions for expressions $e$, replacing the hole with the empty list of statements and erasing statement labels.

\begin{figure}[b]
    \centering
    \small
    \begin{subfigure}{0.45\textwidth}
        \centering
        \begin{align*}
            \mquot{\ell : s' \ssep{} E}\decomp{s} \coloneqq& \ell : s' \ssep{} E\decomp{s} \\
            \mquot{\ell : [] \ssep{} p}\decomp{s} \coloneqq& \ell : s \ssep{} p \\
            \phantom{append()} \\
            \phantom{append()}
        \end{align*}
        \caption{}
    \end{subfigure}
    \begin{subfigure}{0.45\textwidth}
        \centering
        \begin{align*}
            \mquot{\ell : s' \ssep{} E}\expand{\vec{s}} \coloneqq& \ell : s' \ssep{} E\expand{\vec{s}} \\
            \mquot{\ell : [] \ssep{} p}\expand{\vec{s}} \coloneqq& \operatorname{append}(\ell, 1, \vec{s}, p) \\
            \operatorname{append}(\ell, i, (s \ssep{} \vec{s}), p) \coloneqq& \ell, i : s \ssep{} \operatorname{append}(\ell, i+1, \vec{s}, p) \\
            \operatorname{append}(\ell, i, \varnothing, p) \coloneqq& p
        \end{align*}
        \caption{}
    \end{subfigure}
    \caption{The operations of replacing the hole in an evaluation context $E$ with a single statement $E\decomp{\cdot}$ (a), and with a list of statements $E\expand{\cdot}$ (b).}
    \label{fig:semantics:ctxsubst}
\end{figure}

\paragraph{Semantics}
The semantics is defined in Figure~\ref{fig:semantics:rules}. The following paragraphs will discuss each rule in detail, but observe that every rule in the semantics replaces a statement $s$ with a list of statements $\vec{s}$. For the purposes of proving properties, it is helpful to define this as a function $\sreplace{p}{\ell}{s}{\vec{s}}$, which replaces the statement $s$ at $\ell$ in $p$ with the statements $\vec{s}$:
\begin{align*}
    \sreplace{p}{\ell}{s}{\vec{s}} \coloneq E\expand{\vec{s}} \qquad \text{if} \qquad \ctxident{E} = \ell \land p = E\decomp{s}
\end{align*}
Note that ``$\operatorname{replace}$'' preserves label-independence---if $E\decomp{s}$ is label-independent, then so is $E\expand{\vec{s}}$.
We will say a step occurred \emph{at $\ell$}, written $\stepat{p}{}{\ell}{p'}$, if
$p' = \sreplace{p}{\ell}{s}{\vec{s}}$ or $p' = \sreplace{p}{\ell}{s}{\vec{s}}\vsubst{x}{y}$
for some $s$ and $\vec{s}$. We will call $\vec{s}$ the \emph{new statements} and $s$ the \emph{old statement}.

\begin{figure}
    \footnotesize
    \begin{subfigure}{1\textwidth}
        \centering
        \begin{mathpar}
            \inferrule{
            }{\stepreduce{
                E\decomp{\sstmt{y}{x}}
            }{}{
                E\expand{\varnothing}\vsubst{x}{y}
            }} \text{reduce-alias} \and
            \inferrule{
                \mquot{\sfun{f}{x}{\vec{s} \ssep{} y}} \in E \\
                \mquot{\efun{x'}{\vec{s}' \ssep{} y'}} = \freshen{E}{\efun{x}{\vec{s} \ssep{} y}} \\
            }{\stepreduce{
                E\decomp{\sapp{z}{f}{w}}
            }{}{
                E\expand{\sstmt{x'}{w} \ssep{} \vec{s}' \ssep{} \sstmt{z}{y'}}
            }} \text{reduce-call} \and
            \inferrule{
                \mquot{\stuple{x}{x_1, \ldots, x_n}} \in E \\
            }{\stepreduce{
                E\decomp{\sproj{y}{i}{x}}
            }{}{
                E\expand{\sstmt{y}{x_i}}
            }
            } \text{reduce-proj} \and
            \inferrule{
                o \in \{\mquot{\kw{fun} \ldots}, \mquot{\etuple{\ldots}}, \mquot{\eprim{\ldots}}\} \\
                x \notin \freevars{E} \\
            }{\stepreduce{
                E\decomp{\sstmt{x}{o}}
            }{}{
                E\expand{\varnothing}
            }} \text{reduce-gc}
        \end{mathpar}
        \caption{Reduction relation $\stepreduce{p}{}{p'}$, stepping (internally) from $p$ to $p'$.}
        \label{fig:semantics:rules:reduction}
    \end{subfigure} \\ \vspace{1em}
    \begin{subfigure}{1\textwidth}
        \centering
        \begin{mathpar}
            \inferrule{
                \mquot{\sprim{f}{c_f}} \in E\\
                \obsprim{E}{x}{\bar{c}_x} \\
            }{\stepdispatch{
                E\decomp{\sapp{y}{f}{x}}
            }{}{
                E\expand{\stask{y}{c_f}{\bar{c}_x}}
            }} \text{dispatch} \\
            \inferrule{
                \mquot{\sprim{x}{c}} \in E
            }{\obsprim{
                E
            }{x}{
                c
            }}
            \text{pexp-value} \and
            \inferrule
            {
                \mquot{\stuple{x}{x_1, \ldots, x_n}} \in E \\
                \obsprim{E}{x_1}{\bar{c}_1} \\
                \ldots \\
                \obsprim{E}{x_n}{\bar{c}_n} \\
            }
            {\obsprim{
                E
            }{x}{
                (\bar{c}_1, \ldots, \bar{c}_n)
            }}
            \text{pexp-tuple} \and
        \end{mathpar}
        \caption{Dispatch relation $\stepdispatch{p}{}{p'}$, stepping from $p$ to $p'$ by dispatching an external call. The auxiliary ``pexp'' relation $\obsprim{E}{x}{\bar{c}}$ determines whether $x$ corresponds to primitive expression $\bar{c}$ in $E$.}
        \label{fig:semantics:rules:dispatch}
    \end{subfigure} \\ \vspace{1em}
    \begin{subfigure}{1\textwidth}
        \centering
        \begin{mathpar}
            \inferrule{
                \mquot{\vec{s} \ssep{} y} \in \mathcal{C}(c_f, \bar{c}_x) \\
                \mquot{\vec{s}' \ssep{} y'} = \freshen{E}{\vec{s} \ssep{} y} \\
            }{\stepresolve{
                E\decomp{\stask{z}{c_f}{\bar{c}_x}}
            }{}{
                E\expand{\vec{s}' \ssep{} \sstmt{z}{y'}}
            }} \text{resolve} \and
        \end{mathpar}
        \caption{Resolution relation $\stepresolve{p}{}{p'}$, stepping from $p$ to $p'$ by resolving a task to an expression.}
        \label{fig:semantics:rules:resolve}
    \end{subfigure}
    \caption{The operational semantics of \calcname{}, $\stepfull{p}{}{p'}$, stepping from $p$ to $p'$, is defined as the union of $\stepreduce{}{}{}$ (a), $\stepdispatch{}{}{}$ (b), and $\stepresolve{}{}{}$ (c). We write $\stepat{p}{}{\ell}{p'}$ if $p = E\decomp{s}$, $p' = E\expand{\vec{s}}$, and $\ell = \ctxident{E}$, i.e. $(\ell : []) \in E$.}
    \label{fig:semantics:rules}
\end{figure}

\paragraph{Reduction}
Defined in Figure~\ref{fig:semantics:rules}~(a), reduction has four rules. ``reduce-alias'' applies when there is a statement $(\sstmt{y}{x})$ somewhere in the initial expression, and it steps by deleting that statement and replacing all occurrences of $y$ in the rest of the expression with $x$. ``reduce-proj'' applies when there is a statement $(\sproj{y}{i}{x})$, and $x$ is itself defined as a tuple by a statement $(\stuple{x}{x_1, \ldots, x_n})$, and it steps by replacing the projection operation $(\eproj{i}{x})$ with the $i$-th component $x_i$. ``reduce-gc'' applies when there is a statement defining $x$ to be a function, tuple, or primitive, but where $x$ is not referenced elsewhere, and the rule steps by deleting that statement.

The bulk of the work of reduction is done by ``step-call'', an example of which is shown in Figure~\ref{fig:semantics-example}. The rule applies when the initial term has a call statement $(\sapp{z}{f}{w})$ and $f$ is elsewhere defined as a function by a statement $(\sfun{f}{x}{e_f})$, as in Figure~\ref{fig:semantics-example}~(a). First, the function definition $(\efun{x}{e_f})$ is freshened with respect to the evaluation context $E$, in (b). We call the fresh parameter $x'$, the fresh statement list $\vec{s}'$, and the fresh return variable $y'$ (in the example, $\kw{x0}$, $(\sapp{\kw{y0}}{\kw{f}}{\kw{x0}})$, and $\kw{y0}$ respectively). Finally, the call statement is replaced by a statement assigning the fresh parameter $x'$ to the argument $w$, followed by the fresh statement list, followed by a statement assigning the call variable $z$ to the fresh return variable $y'$, in (c). Freshening ensures that variable names remain unique and the hole-filling operation preserves label independence.

\begin{figure}
    \footnotesize
    \hfill
    \begin{subfigure}[t]{0.2\textwidth}
        \centering
        \begin{minted}{python}
1: g := fun x:
       y := f x
       y
        
2: z1 := g z0
3: z2 := g z1



   z2
        \end{minted}
        \caption{Before, $p$}
        \label{fig:semantics-example:a}
    \end{subfigure}\hfill
    \begin{subfigure}[t]{0.2\textwidth}
        \centering
        \begin{minted}[escapeinside=``]{python}
`\ `





fun `\textbf{x0}`:
  `\textbf{y0}` := f `\textbf{x0}`
  `\textbf{y0}`
`\ `
        \end{minted}
        \caption{Fresh copy of the definition of $\kw{g}$}
        \label{fig:semantics-example:b}
    \end{subfigure}\hfill
    \begin{subfigure}[t]{0.2\textwidth}
        \centering
        \begin{minted}[escapeinside=``]{python}
1:   g := fun x:
         y := f x
         y
  
2:   z1 := g z0
`\sout{3:   z2 := g z1}`
`\textbf{3,1: x0 := z1}`
`\textbf{3,2: y0 := f x0}`
`\textbf{3,3: z2 := y0}`
     z2
        \end{minted}
        \caption{Statement replacement at label $3$}
        \label{fig:semantics-example:c}
    \end{subfigure}\hfill
    \begin{subfigure}[t]{0.2\textwidth}
        \centering
        \begin{minted}[escapeinside=``]{python}
1:   g := fun x:
         y := f x
         y
  
2:   z1 := g z0

3,1: x0 := z1
3,2: y0 := f x0
3,3: z2 := y0
     z2
        \end{minted}
        \caption{After, $p'$}
        \label{fig:semantics-example:c}
    \end{subfigure}
    \caption{An illustration of the ``reduce-call'' rule for $\stepat{p}{}{3}{p'}$, stepping from $p$ to $p'$ at label $3$.}
    \label{fig:semantics-example}
\end{figure}

\paragraph{Dispatch}
Defined in Figure~\ref{fig:semantics:rules}~(b), dispatch has one rule, as well as an auxiliary ``pexp'' relation. External functions are assumed to accept only on primitive values, not \calcname{} syntax, as arguments. Because they may expect more than one argument, we have a notion of a primitive expression $\bar{c}$ (recall that this is either a primitive $c$ or a tuple of primitive expressions). However, we need to be able to convert ordinary \calcname{} tuples into primitive expressions. This is accomplished by $\obsprim{E}{x}{\bar{c}}$, which says that in a context $E$, variable $x$ corresponds to $\bar{c}$. If $x$ is defined to be a primitive $c$ (``pexp-value''), then it corresponds to the primitive expression containing just $c$. If instead $x$ is defined to be a tuple of $x_i$, (``pexp-tuple'') then it corresponds to the primitive expression that is a tuple, where the $i$-th component is the primitive expression corresponding to $x_i$. The dispatch semantics, ``dispatch'' applies when there is a call $(\sapp{y}{f}{x})$ where $f$ is bound to a primitive $\eprim{c_f}$ and $x$ corresponds to a primitive expression $\bar{c}_x$, and replaces the call statement with a task statement $(\stask{y}{c_f}{\bar{c}_x})$.

For example, in an evaluation context containing
\begin{center}
    \kw{f := <cf>; w := <c1>; x := <c2>; y := (w, x)}
\end{center}
the statement \kw{z := f y} steps via dispatch to \kw{z := <<cf (c1, c2)>>}. It is at this step that an actual runtime would begin executing $\kw{cf (c1, c2)}$, e.g. by making a network request.

\paragraph{Resolution}
Because tasks model opaque external calls, their semantics must be provided to the runtime. We model this with an \emph{external semantics} relation $\mathcal{C}(c_f, \bar{c}_x) \subseteq S_e$, where $S_e$ is the set of expressions that are well-formed in the empty context. $c_f$ is the primitive function that was called, and $\bar{c}_x$ is the primitive (or tuple thereof) passed as the call argument. Crucially, the external semantics yields a \emph{set} of expressions, allowing external calls to be nondeterministic, but this set is determined only by the function called and the argument passed, encoding our core assumption that the behavior of external calls is determined by their argument. For deterministic external calls, such as primitive arithmetic, the set will be a singleton, e.g. $\mathcal{C}(\kw{+}, \kw{(3,4)}) = \{\kw{<7>}\}$. However, a coin flip can be modeled as a nondeterministic boolean, returning $\eprim{\kw{f}}$ or $\eprim{\kw{t}}$, i.e. $\mathcal{C}(\kw{coin}, \kw{()}) = \{\eprim{\kw{f}},\eprim{\kw{t}}\}$. Further, external calls need not just return primitives, and can affect control-flow by returning \calcname{} functions, which will be utilized in Section~\ref{sec:epiclang:data}.

Defined in Figure~\ref{fig:semantics:rules}~(c), resolution has one rule. It applies when there is a task statement $(\stask{z}{c_f}{\bar{c}_x})$, and it steps by nondeterministically choosing an expression $\vec{s}\ssep{} y$ from the external semantics for the task $\mathcal{C}(c_f, \bar{c}_x)$, freshening it to the expression $\vec{s}' \ssep{} y'$, and then replacing the statement with the fresh statements $\vec{s}'$ and a statement $\sstmt{z}{y'}$ assigning the variable $z$ that was bound to the task to the fresh return variable $y'$.

For example, \kw{z := <<cf (c1, c2)>>} might step to \kw{y := <c3>; z := y}. An actual runtime implementation would wait to take this step until the external task had finished and yielded $\vec{s} \ssep{} y$, e.g. by receiving a network response.

\paragraph{Deterministic Semantics}
The semantics $\stepfull{}{}{}$ are not confluent \emph{per se}, due to nondeterminism in the ``resolve'' rule: e.g.
$\etask{\kw{coin}}{()}$ could step to $\eprim{\kw{f}}$ or $\eprim{\kw{t}}$, which have no common reduct. Thus, we define a \emph{deterministic resolution} relation $\stepresolve{p}{T}{p'}$ (Figure~\ref{fig:semantics:rules:resolvedet}) and a \emph{deterministic semantics} $\stepfull{p}{T}{p'}$ (the union of reduction, dispatch, and \emph{deterministic} resolution). Here, $T$ is an \emph{external interaction trace}, a partial function from statement labels $\ell$ to expressions $e$. Instead of picking the call result nondeterministically like ``resolve'', ``resolvedet'' uses the identifier $\ell$ of the task statement being stepped to look up the result specified in $T$.

\begin{proposition}
    $p_0 \Rightarrow \ldots \Rightarrow p_n$ iff there exists $T$ such that $p_0 \Rightarrow_T \ldots \Rightarrow_T p_n$.
\end{proposition}
\begin{proof}
    The forward direction holds, because ``resolve'' can always choose whatever $(\vec{s} \ssep{} y) \in \mathcal{C}(c_f, \bar{c}_x)$ ``resolvedet'' did. The backward direction holds because the nondeterministic choices made by ``resolve'' can be collected to form a $T$. Crucially, no sequence of steps will ever step the same $\ell$ twice, so no two ``resolve'' steps can require $T$ to map $\ell$ to two different expressions. No sequence can step the same $\ell$ twice because if there is a step at $\ell$, the $\ell$ and its statement are removed. No subsequent step can recreate the label $\ell$ due to label-independence.
\end{proof}

\begin{figure}
    \centering
    \begin{mathpar}
        \inferrule{
            \ell = \ctxident{E} \\
            \mquot{\vec{s} \ssep{} y} = T(\ell) \\
            \mquot{\vec{s} \ssep{} y} \in \mathcal{C}(c_f, \bar{c}_x) \\
            \mquot{\vec{s}' \ssep{} y'} = \freshen{E}{\vec{s} \ssep{} y} \\
        }{\stepresolve{
            E\decomp{\stask{z}{c_f}{\bar{c}_x}}
        }{T}{
            E\expand{\vec{s}' \ssep{} \sstmt{z}{y'}}
        }} \text{resolvedet} \and
    \end{mathpar}
    \caption{Deterministic resolution relation $\stepresolve{p}{T}{p'}$, stepping from $p$ to $p'$ by resolving a task to an expression according to environment trace $T$. The full deterministic semantics of \calcname{}, $\stepfull{}{T}{}$ is the union of $\stepreduce{}{}{}$, $\stepdispatch{}{}{}$, and $\stepresolve{}{T}{}$. We write $\stepat{p}{T}{\ell}{p'}$ to mean that the step occurred at $\ell$.}
    \label{fig:semantics:rules:resolvedet}
\end{figure}

\paragraph{Properties} We now prove several important properties of the semantics.
\begin{proposition}[Determinacy]
    For any $T, \ell$, and well-formed $p$, there is at most one $p'$ such that $\stepat{p}{T}{\ell}{p'}$.
\end{proposition}
\begin{proof}
    We can show that, for each kind of old statement, stepping is deterministic. The only interesting cases are $\sproj{y}{i}{x}$ and $\sapp{y}{f}{x}$. In them, a statement defining $x$ or $f$ (respectively) is looked for in $E$. However, if such a statement is found, it must be unique, since $p$ is well-formed and well-formedness prevents re-binding of variable names.
\end{proof}

\begin{proposition}[Preservation]
    If $\Gamma \vdash p$ and $\stepfull{p}{T}{p'}$, then $\Gamma \vdash p'$.
\end{proposition}
\begin{proof}
    We must show that $p'$ is both well-formed as an expression and identifier-independent.

    Identifier-independence holds because, for every rule except ``reduce-alias'', $E' = \sreplace{E}{\ell}{s}{\vec{s}}$, and replacement preserves independence. For ``reduce-alias'', $E' = \sreplace{E}{\ell}{s}{\vec{s}}\vsubst{x}{y}$, but variable substitution does not affect statement identifiers.

    For expression well-formedness, we must consider whether removing the old statement causes problems for its bound variable $y$, and whether any new statements reference unbound variables. For the former, note that every rule provides a new statement that re-binds $y$, except ``reduce-gc'' which ensures $y$ does not occur elsewhere, and ``reduce-alias'' which replaces $y$ elsewhere. For the latter, inspection of each rule shows that variables referenced in the new statements must be bound at the place they are inserted.
\end{proof}

\begin{proposition}[Statement Replacement Commutes]
    \label{prop:replacement-commutes}
    If $p$ is well-formed, then 
    $$\sreplace{\sreplace{p}{\ell_1}{s_1}{\vec{s}_1}}{\ell_2}{s_2}{\vec{s}_2} = \sreplace{\sreplace{p}{\ell_2}{s_2}{\vec{s}_2}}{\ell_1}{s_1}{\vec{s}_1}$$
\end{proposition}
\begin{proof}
    The labels $\ell_1,\ell_2$ must be independent (or else one side of the equality would be undefined). Thus the replacement operations act on disjoint parts of $p$, and so both sides are equal to doing the replacements simultaneously.
\end{proof}

\begin{proposition}[Strong Confluence]
    \label{prop:confluence}
    For any $T, p$, and $p_1 \neq p_2$, if $\stepfull{p}{T}{p_1}$ at $\ell_1$ and $\stepfull{p}{T}{p_2}$ at $\ell_2$, then there exists $p_{12}$ such that $\stepfull{p_1}{T}{p_{12}}$ at $\ell_2$ and $\stepfull{p_2}{T}{p_{12}}$ at $\ell_1$.
\end{proposition}
\begin{proof}
    Since every rule except ``reduce-alias'' simply does a statement replacement, which commutes (Proposition~\ref{prop:replacement-commutes}), every rule except ``reduce-alias'' commutes. To see that every other rule commutes with ``reduce-alias'' for $\sstmt{y}{x}$, consider each rule, and note that generating the new statements commutes with replacing $y$ with $x$. To see that ``reduce-alias'' commutes with itself, note that either the two alias statements have disjoint variables, in which case variable substitution commutes, or they are of the form $\sstmt{y}{x}$, $\sstmt{z}{y}$, in which case either way all occurrences of $y$ and $z$ will be replaced with $x$.
\end{proof}

\subsection{Evaluation Strategy}
\label{sec:corecalculus:eval}
The confluence of our semantics enables us to soundly step statements in unusual orders. However, for an implementation, we need an actual evaluation strategy. Sequential, call-by-value evaluation would step statements in order. Lazy evaluation would work backward from the return variable, stepping statements needed to turn it into a value. By contrast, the idea of opportunistic evaluation is to step as many statements as we can, anywhere we can. Some may not be able to step yet, e.g. a function call where the function does not yet have a known definition, but we will step all the rest.

For a well-formed labeled expression $p$ and a set of labels $L$, we say $p$ \emph{steps simultaneously at $L$} to $p$ (with environment trace $T$), $p \Rightarrow^L_T p'$, if stepping each identifier $\ell \in L$ in some order yields $p'$. Note that, by confluence, the order in which each individual step is taken does not matter, so there is at most one such $p'$. We say $\operatorname{steppable}(p,\ell)$ if there exists a $p'$ such that $\stepat{p}{}{\ell}{p'}$, and we let $\operatorname{steppable}(p) = \{\ell : \operatorname{steppable}(p, \ell)\}$. Thus we can finally define \emph{opportunistic evaluation} $p \Rrightarrow_T p'$ to be stepping $p$ at every steppable label simultaneously. These operations are defined formally in Figure~\ref{fig:semantics:oppeval}. For an example of opportunistic evaluation, see Figure~\ref{fig:reduction-example}.

\begin{figure}
    \centering
    \begin{mathpar}
        \inferrule{
            L = \{\ell_1, \ldots, \ell_n\} \\
            \forall i\leq n.\ \stepat{p_{i-1}}{T}{\ell_i}{p_i} \\
        }{
            \stepat{p_0}{T}{L}{p_n}
        } \text{} \and
        \inferrule{
            \stepat{p}{}{\ell}{p'}
        }{
            \operatorname{steppable}(p, \ell)
        } \text{} \and
        \inferrule{
            L = \left\{\ell : \operatorname{steppable}(p, \ell) \right\} \\
            p \Rightarrow^L_T p' \\
        }{p \Rrightarrow_T p'} \text{} \and
    \end{mathpar}
    \caption{Stepping at set of labels $L$ simultaneously, $\stepat{p}{T}{L}{p'}$. Steppability of a label $\ell$ in $p$, $\operatorname{steppable}(p, \ell)$. Opportunistic evaluation, $p \Rrightarrow_T p'$. $T$ is an environment interaction trace.}
    \label{fig:semantics:oppeval}
\end{figure}

\begin{figure}[t]
    \footnotesize
    \begin{subfigure}[t]{0.15\textwidth}
        \centering
        \begin{minted}[escapeinside=``]{python}
1: g := fun x:
         y := f x
         y
`\textbf{2: z1 := g z0}`



`\textbf{3: z2 := g z1}`



   z2
        \end{minted}
        \caption{$p_0$}
        \label{fig:reduction-example:a}
    \end{subfigure}\hfill
    \begin{subfigure}[t]{0.15\textwidth}
        \centering
        \begin{minted}[escapeinside=``]{python}
1:   g := fun x:
         y := f x
         y
`\sout{2:   z1 := g z0}`
`\textbf{2,1: x0 := z0}`
`\textbf{2,2: y0 := f x0}`
`\textbf{2,3: z1 := y0}`
3:   z2 := g z1



     z2
        \end{minted}
        \caption{$\stepat{p_0}{}{2}{p_1}$}
        \label{fig:reduction-example:b}
    \end{subfigure}\hfill
    \begin{subfigure}[t]{0.175\textwidth}
        \centering
        \begin{minted}[escapeinside=``]{python}
1:   g := fun x:
         y := f x
         y

2,1: x0 := z0
2,2: y0 := f x0
2,3: z1 := y0
`\sout{3:   z2 := g z1}`
`\textbf{3,1: x1 := z1}`
`\textbf{3,2: y1 := f x1}`
`\textbf{3,3: z2 := y1}`
     z2
        \end{minted}
        \caption{$\stepat{p_1}{}{3}{p_2}$}
        \label{fig:reduction-example:c}
    \end{subfigure}\hfill
    \begin{subfigure}[t]{0.175\textwidth}
        \centering
        \begin{minted}[escapeinside=``]{python}
1:   g := fun x:
         y := f x
         y

2,1: x0 := z0
2,2: y0 := f x0
2,3: z1 := y0

3,1: x1 := z1
3,2: y1 := f x1
3,3: z2 := y1
     z2
        \end{minted}
        \caption{$p_2$}
        \label{fig:reduction-example:d}
    \end{subfigure}
    \caption{A single step of opportunistic evaluation, $p_0 \Rrightarrow p_2$. The steppable statements $L$ in $p_0$ are $2$ and $3$, shown in bold in (a). The statement labelled $1$ cannot step because $g$ is referenced elsewhere in the term. One concrete sequence of semantics steps from $p_0$ to $p_2$ is $p_0 \Rightarrow^2 p_1 \Rightarrow^3 p_2$, shown in (b) and (c).}
    \label{fig:reduction-example}
\end{figure}

A key property of opportunistic evaluation is that, for any set of labels $L$ that could be stepped by the semantics, they will all eventually be stepped by opportunistic evaluation. We call this \emph{fairness}, by analogy with fair scheduling strategies. Intuitively, fairness holds because, if we could step at a sequence of labels $\ell_1, \ldots, \ell_n$, and instead we take $n$ steps of opportunistic evaluation, by the $i$-th step of opportunistic evaluation we are guaranteed to have stepped $\ell_i$ (though each opportunistic evaluation step may be many individual semantics steps, and so $\ell_i$ may have been incidentally stepped earlier).

Note that sequential evaluation is not fair: the program \kwtt{() = diverge (); () = print <"foo">} will never execute \kwtt{print}, because it will loop forever on \kw{diverge}, even though the semantics allows stepping \kwtt{print} at any time.

\begin{proposition}[Fairness]
    For any $p, T, L$, if $\stepat{p}{T}{L}{p'}$ for some $p'$, then $p \Rrightarrow_T^* p''$ for some $p''$ such that $\stepat{p}{T}{L_*}{p''}$ and $L \subseteq L_*$.
\end{proposition}
\begin{proof}
    We do induction on the number of steps $\lvert L \rvert = n$. For $n=1$, $L = \{\ell\}$. If we take one step of opportunistic evaluation, we get $L_* = \operatorname{steppable}(p)$. But $\ell$ must have been steppable, so $L \subseteq L_*$.

    For $n>1$, let $L = \{\ell\} \cup L'$ for some $\ell$ and $L'$. By the induction hypothesis, we have an $L'_*$ stepped by opportunistic evaluation such that $\stepat{p}{T}{L'_*}{p''}$ for some $p''$ and $L' \subseteq L'_*$. If $\ell \in L'_*$, then already $L \subseteq L'_*$, and we are done. If $\ell \notin L'_*$, let $L^\Delta = L'_* \setminus L'$, i.e. all the extra steps done by opportunistic evaluation. Since $\ell$ is steppable at $p'$ and $\stepat{p'}{T}{L^\Delta}{p''}$, by Proposition~\ref{prop:steppable-stable}, $\ell$ is steppable at $p''$. Thus we can take one more step of opportunistic evaluation, $p'' \Rrightarrow_T p'''$ and let $L_*$ be such that $\stepat{p}{T}{L_*}{p'''}$.
\end{proof}

\begin{proposition}
    \label{prop:steppable-stable}
    If $\stepat{p}{T}{L}{p'}$, $\ell \in \operatorname{steppable}(p)$, and $\ell \notin L$, then $\ell \in \operatorname{steppable}(p')$.
\end{proposition}
\begin{proof}
    We do induction on the number of steps $\lvert L \rvert = n$. For $n>1$, let $L = \{\ell'\} \cup L'$. We have $\stepat{p}{T}{L'}{p''}$ and $\stepat{p''}{T}{\{\ell'\}}{p'}$, so applying the induction hypothesis twice, we get that $\ell$ steppable in $p''$ and then $\ell$ steppable in $p'$. For $n=1$, $L = \{\ell'\}$. Let $s$ and $s'$ be the statements at $\ell$ whose steppability is being judged before and after stepping, i.e. $E\decomp{s} = p$, $E'\decomp{s'} = p'$, and $\ctxident{E} = \ctxident{E'} = \ell$. Consider the possible statements $s$ could be. If $s$ is an alias, it is unconditionally steppable. If $s$ is a function call, the only way it could lose steppability is if the statement defining $f$ were deleted by ``step-gc'', but that can't have happened, since $f$ is referenced by $s$. Similarly, if $s$ is a projection, the statement defining $x$ cannot have been deleted. If $s$ is a function, tuple, or primitive, then by steppability, $x \notin \freevars{E}$, but then there's no way stepping at $\ell'$ could have introduced an occurrence of $x$, so it's still steppable. If $s$ is a task, then it is unconditionally steppable.
\end{proof}

\begin{proposition}[Soundness]
    For terminating programs, opportunistic and sequential evaluation are equivalent.
\end{proposition}
\begin{proof}
Note that sequential evaluation, when it terminates, executes every statement. Opportunistic evaluation also executes every statement when it terminates. Thus by confluence, they make exactly the same external calls and produce the same result.
\end{proof}

In fact, for nonterminating computations, opportunistic evaluation can produce more external calls than both strict and lazy evaluation combined. For example, in
\begin{inlinebox}
    \begin{minipage}[c]{0.45\textwidth}
\begin{verbatim}
() := diverge (); () := print <"foo">; <"bar">
\end{verbatim}
    \end{minipage}
\end{inlinebox}
neither strategy would print ``foo''---strict evaluation because it would get stuck stepping calls to \kwtt{diverge ()} forever, and lazy evaluation because printing is not needed for the return value \kwtt{<"bar">}. However, opportunistic evaluation would perform the print.

\section{\toolname{} Language}
\label{sec:epiclang}

As a core calculus, \calcname{} lacks standard features such as data structures and control-flow constructs.
Thus, we build \toolname{} on top of \calcname{}; \toolname{}
offers data structures and standard control-flow (Section~\ref{sec:epiclang:data}), as well as automatic streaming (Section~\ref{sec:epiclang:streaming}). We also show how sequencing can be enforced between external calls such as printing, where ordering is important (Section~\ref{sec:epiclang:sequencing}).

\subsection{Data Structures and Control-Flow via Church Encodings}
\label{sec:epiclang:data}

Rather than changing the core calculus to add explicit data-constructors or control-flow, we define them in-language using Church encodings~\cite{churchencoding}. This has two advantages: first, they are already supported by the core calculus and thus automatically benefit from opportunistic evaluation; second, they offer a natural notion of ``partial data'', discussed in Section~\ref{sec:epiclang:streaming}. Though Church encodings are an established concept, we review them here.
\paragraph{Booleans}
The Church encodings of \kwtt{if}, \kw{true}, and \kw{false} are:
\begin{inlinebox}
    \hfill
    \begin{minipage}[c]{0.40\textwidth}
\begin{verbatim}
if := fun (cond, case_true, case_false):
    cond (case_true, case_false)
\end{verbatim}
    \end{minipage}
    \hfill
    \begin{minipage}[c]{0.20\textwidth}
\begin{verbatim}
false := fun (t, f):
    f ()
\end{verbatim}
    \end{minipage}
    \hfill
    \begin{minipage}[c]{0.19\textwidth}
\begin{verbatim}
true := fun (t, f):
    t ()
\end{verbatim}
    \end{minipage}
    \hfill\ 
\end{inlinebox}
As an example, if the condition of an \kwtt{if} is \kw{false}, the false branch will be executed:
\begin{inlinebox}
    \footnotesize
    \hfill
    \begin{minipage}[c]{0.20\textwidth}
    \begin{minted}[escapeinside=||]{text}
c := if (false, a, b)
    \end{minted}
    \end{minipage}
    \hfill
    $\Rightarrow^*$
    \hfill
    \begin{minipage}[c]{0.16\textwidth}
    \begin{minted}[escapeinside=||]{python}
c := false (a, b)
    \end{minted}
    \end{minipage}
    \hfill
    $\Rightarrow^*$
    \hfill
    \begin{minipage}[c]{0.06\textwidth}
    \begin{minted}[escapeinside=||]{python}
c := b ()
    \end{minted}
    \end{minipage}
    \hfill\ 
\end{inlinebox}
Note that external tasks can affect program control flow by returning these Church Booleans.

\paragraph{Lists}
Lists are Church encoded similar to \kw{fold}---as functions taking an initial ``accumulator state'' (\kw{nil}) and a function (\kw{cons}) which is called once per element of the list, with the accumulator state passed through. For instance, the Church encoding of the list \kw{[a,b,c]} is:
\begin{inlinebox}
    \hfill
    \begin{minipage}[c]{0.40\textwidth}
\begin{verbatim}
list_abc := fun (nil, cons):
    cons (cons (cons (hil, a), b), c)
\end{verbatim}
    \end{minipage}
    \hfill
    $\approx$
    \hfill
    \begin{minipage}[c]{0.30\textwidth}
\begin{verbatim}
list_abc := fun (nil, cons):
    state := nil
    state := cons (state, a)
    state := cons (state, b)
    state := cons (state, c)
    state
\end{verbatim}
    \end{minipage}
    \hfill\ 
\end{inlinebox}

Using this intuition, we can define operations for constructing (\kw{nil}, \kw{cons}) and folding over (\kw{fold}), and concatenating (\kw{concat}) lists:
\begin{inlinebox}
    \hfill
    \begin{minipage}{0.4\textwidth}
\begin{minted}{text}
nil := fun (state, append):
    state

cons := fun (hd, tl):
    fun (state, append):
        state := append (state, hd)
        state := fold (tl, state, append)
        state
\end{minted}
    \end{minipage}
    \hfill
    \begin{minipage}{0.4\textwidth}
\begin{minted}{text}
fold := fun (l, init, each):
    l (init, each)

concat := fun (l1, l2):
    fun (state, append):
        state := fold (l1, state, append)
        state := fold (l2, state, append)
        state
\end{minted}
    \end{minipage}
    \hfill\
\end{inlinebox}

With these lists, opportunistic evaluation automatically unrolls loops, which is crucial for achieving parallelism. For \kw{list_abc} defined above:
\begin{inlinebox}
    \footnotesize
    \hfill
    \begin{minipage}[c]{0.26\textwidth}
    \begin{minted}[escapeinside=||]{python}
result := fold (
    list_abc, init, each)
    \end{minted}
    \end{minipage}
    \hfill
    $\Rightarrow^*$
    \hfill
    \begin{minipage}[c]{0.20\textwidth}
    \begin{minted}[escapeinside=||]{python}
result := list_abc (
    init, each)
    \end{minted}
    \end{minipage}
    \hfill
    $\Rightarrow^*$
    \hfill
    \begin{minipage}[c]{0.24\textwidth}
    \begin{minted}[escapeinside=||]{python}
state := init
state := each (state, a)
state := each (state, b)
state := each (state, c)
result := state
    \end{minted}
    \end{minipage}
    \hfill\ 
\end{inlinebox}
At this point, all three calls to \kw{each} are at the top level, so they will be executed opportunistically (if \kw{each} is defined), regardless of whether \kw{init} is.

Note that it is also possible to define recursive functions via a standard fixpoint combinator.

\subsection{Partial Data and Streaming Computation}
\label{sec:epiclang:streaming}
Not only do Church encodings provide a way of defining instances of traditional algebraic data types, like lists, they also provide a way of defining \emph{partial} versions of these structures. For instance, the partial list $(\kw{["H"] + h1 + ["J"] + h2})$  (i.e. \kw{"H"} at the beginning and \kw{"J"} somewhere after) can be represented as follows, with the free variables \kw{h1} and \kw{h2} standing in for the unknown parts of the list:
\begin{inlinebox}
    \hfill
    \begin{minipage}[c]{0.3\textwidth}
\begin{minted}{text}
partial := fun (state, append):
    state := append (state, "H")
    state := h1     (state, append)
    state := append (state, "J")
    state := h2     (state, append)
    state
\end{minted}
    \end{minipage}
    \hfill\ 
\end{inlinebox}
Note that this is richer than typical notions of streaming, which restrict partial lists to only have a hole at the end. However, if we take two partial lists and concatenate them (e.g. concatenating two LLM calls), even if individually they only had a hole at the end, the result will have two holes, with one in the middle. With traditional data constructors and pattern matching, it is not obvious how to represent data with holes ``in the middle'' like this. However, these Church-encoded partial lists behave the way we want, with loops over them unrolling automatically, exposing loop iterations for all known list element.
\begin{inlinebox}
    \hfill
    \begin{minipage}[c]{0.22\textwidth}
\begin{minted}[escapeinside=||]{python}
for v in partial:
    print (&stdout, v)
\end{minted}
    \end{minipage}
    \hfill
    $\approx$
    \hfill
    \begin{minipage}[c]{0.31\textwidth}
\begin{verbatim}
stdout := partial (stdout, print)
\end{verbatim}
    \end{minipage}
    \hfill
    $\Rightarrow^*$
    \hfill
    \begin{minipage}[c]{0.29\textwidth}
\begin{minted}{text}
stdout := print (stdout, "H")
stdout := h1    (stdout, print)
stdout := print (stdout, "J")
stdout := h2    (stdout, print)
\end{minted}
    \end{minipage}
    \hfill\ 
\end{inlinebox}

This idea can be leveraged to achieve automatic streaming of programs in \toolname{}, evaluating on their partial output as it becomes available, without waiting for the entire result. We give an example of such a streaming external call in Figure~\ref{fig:streaming-example}. Note that a similar approach can be used to stream richer data than lists, such trees and other algebraic datatypes.

\begin{figure}[t]
    \centering
    \scriptsize
    \vskip 0pt
    \hfill
    \begin{subfigure}[t]{0.27\textwidth}
\begin{minted}{text}
resp := <<cf "foo">>








stdout := resp (stdout, print)
\end{minted}
        \caption{}
    \end{subfigure}
    \hfill
    \begin{subfigure}[t]{0.31\textwidth}
\begin{minted}{text}
resp := fun (state, append):
    state := append (state, "a")
    h1    := <<cf1 ()>>
    state := h1     (state, append)
    state := append (state, "m")
    h2    := <<cf2 ()>>
    state := h2     (state, append)
    state := append (state, "z")
    state
stdout = resp (stdout, print)
\end{minted}
        \caption{}
    \end{subfigure}
    \hfill
    \begin{subfigure}[t]{0.27\textwidth}
\begin{minted}{text}

stdout := print (stdout, "a")
h1     := <<cf1 ()>>
stdout := h1    (stdout, print)
stdout := print (stdout, "m")
h2     := <<cf2 ()>>
stdout := h2    (stdout, print)
stdout := print (stdout, "z")


\end{minted}
        \caption{}
    \end{subfigure}
    \hfill\ 
    \caption{A term that loops over and prints a list of strings from an unresolved external call (a). The task $\etask{\kwtt{cf}}{\kwtt{"foo"}}$ might resolve to the partial string \kw{["a"] + h1 + ["m"] + h2 + ["z"]}, where the holes \kw{h1}, \kw{h2} are new unresolved tasks $\etask{\kw{cf1}}{()}$ and $\etask{\kw{cf2}}{()}$ (b). The expression could step then step to (c), where opportunistic evaluation could begin executing the \kw{"a"}, \kw{"m"}, and \kw{"z"} iterations of the loop.}
    \label{fig:streaming-example}
\end{figure}

\subsection{Enforcing Sequencing of External Calls}
\label{sec:epiclang:sequencing}

\calcname{} makes the foundational assumption that dependencies between external calls can be fully determined by their arguments, which licenses us to execute calls with no data-dependency in parallel. However, many operations, such as mutating filesystem state or printing to standard output are sensitive to the order in which they execute, and usually do not encode these dependencies in their arguments.

We address this by reifying
logical orderings as explicit data dependencies:
rather than a 
signature \kwtt{() := print (stdout, s)} for \kwtt{print} that takes a file descriptor \kw{stdout} and a string \kw{s}, its signature is \kwtt{stdout1 := print (stdout, s)}, allowing subsequent operations to take \kw{stdout1}, thus data-depending on the \kwtt{print} call and executing in order.

This approach can simulate an execution where all effects are totally ordered by threading a 
handle through every external call like the example below.
\begin{inlinebox}
    \begin{minipage}[c]{0.47\textwidth}
\begin{minted}{text}
   (thread, fd1) := open  (thread, <"foo.txt">)
   (thread, c)   := read  (thread, fd1)
   (thread, fd2) := open  (thread, <"bar.txt">)
    thread       := write (thread, fd2, c)
\end{minted}
    \end{minipage}
\end{inlinebox}
At the same time, it is not limited to sequential ordering but can be used to encode arbitrary partial order dependencies.
Fork/join concurrency can be encoded by adding external calls
external calls
\begin{inlinebox}
    \kw{(thread, thread2) := fork (thread)} \qquad\qquad and \qquad\qquad \kw{thread =: join (thread, thread2)}.
\end{inlinebox}
For example:
\begin{inlinebox}
    \begin{minipage}[c]{0.51\textwidth}
\begin{minted}{text}
1: (thread, fd1)     := open  (thread, <"foo.txt">)
2: (thread, thread2) := fork  (thread)
3: (thread2, c)      := read  (thread2, fd1)
4: (thread, fd2)     := open  (thread, <"bar.txt">)
5:  thread           := join  (thread, thread2)
6:  thread           := write (thread, fd2, c)
\end{minted}
    \end{minipage}
\end{inlinebox}
Statements $3$ and $4$ have no data-dependency, and so can execute in either order, but statement $1$ must execute first and $6$ must execute last.

\section{Implementation}
\label{sec:impl}

We develop an implementation of \toolname{} as a library in Python, supporting writing programs in Python syntax and marking specific Python functions, such as calls to machine learning libraries, as external calls.
In this section, we expand on two key aspects of our implementation: (1) task management, and (2) evaluation efficiency.

\paragraph{Task Management}
The set $C$ of primitives is the set of Python objects in the runtime---that is, syntactic $\eprim{c}$ statements contain references to arbitrary Python objects. Python objects referenced by ``call'' statements are assumed either to be ordinary Python functions or Python ``async coroutines''.
When evaluation encounters a ``dispatch'' to $c_f$ with argument $\bar{c}_x$, first it checks whether $c_f$ is a function or a coroutine. If it is a function, the runtime immediately executes $c_f$ with its arguments, and its result is placed into the term immediately, executing synchronously and skipping task creation and resolution altogether. This is useful for simple external calls like manipulation of Python strings that would not benefit from parallelization.
When instead $c_f$ is a coroutine, it is scheduled as a task with Python's async machinery and added to a global set of unresolved tasks. The task reference is then stored in a ``task'' statement (which is a slight deviation from the syntax in Section~\ref{sec:corecalculus:syntax}, for ease of implementation). There is no syntactic facility for users to define tasks---they can only be created by ``dispatch''.
This ensures that every task statement refers to some task that has been dispatched by the runtime.
When evaluation encounters a ``task'' statement, it checks whether the Python async task has a value. If so, the value is assumed to be an expression, and the task is resolved, replacing the ``task'' statement with the statements in the expression.
The evaluation loop proceeds by repeatedly running evaluation. If evaluation changes the term, evaluation continues. If the term does not change, it checks whether there are any outstanding tasks. If so, it sleeps until one of them has completed and continues evaluating. If the term does not change and there are no outstanding tasks, evaluation terminates.

\paragraph{Evaluation Efficiency}
While $p \Rightarrow p'$ can be implemented according to its presentation in Section~\ref{sec:corecalculus:eval}, this is too inefficient to be practical. In particular, doing a linear scan over statements to identify steppable statements, plus another linear scan to do variable substitution, causes a large blow-up in runtime as a function of term size. This becomes a problem when handling moderately sized Church-encoded strings (e.g. as returned from an LLM), since each character requires a separate statement.
Fortunately, there are efficient solutions: rather than representing the term as a list of statements with simple variable identifiers, we represent it as a graph with two kinds of vertices: one for statements and one for variables. We then maintain bidirectional mappings between statements and variables. Then, when replacing a statement or renaming a variable, identification and mutation of affected statements is efficient.

\section{Evaluation}
\label{sec:evaluation}
We evaluate our approach, focusing on answering the following key research questions (RQs):
\begin{itemize}
\item \textbf{RQ1:} How does \toolname{} compare to practitioner-style Python implementations without explicit parallelization or streaming?
\item \textbf{RQ2:} How does \toolname{} compare to existing frameworks for parallelization and streaming (Apache Flink, bash, PaSh, and SGLang)?
\item \textbf{RQ3:} What is the overhead of \toolname{} compared to hand-optimized parallel and streaming Rust  implementations?
\item \textbf{RQ4:} How much does opportunistic evaluation improve the performance of \toolname{} over call-by-value evaluation?
\end{itemize}

\subsection{Experimental Design}
\label{exp-design}

\paragraph{Benchmarks.} We evaluated \toolname{} using 
five benchmark programs that are representative of real-world LLM scripting tasks. 
\loopbench{} is the motivating example from Section~\ref{sec:motivating}.
\pipelinebench{} concatenates several independent calls to a language model and feeds them to an external text-to-speech tool in a streaming fashion.
\toolusebench{} implements a fact-checker via ``tool use''~\cite{schick2024toolformer} iteratively composing calls to a language model with a Wikipedia lookup tool that fetches information through Wikipedia's API.
\totbench{} is an implementation of Tree-of-Thoughts~\cite{tree-of-thoughts}, 
which treats reasoning as a graph search problem, where the node expansion and scoring functions are both implemented using LLMs.
\jsonbench{} extends simple language model calls using constrained decoding~\cite{poesia2022synchromesh}, enforcing that the LLM output satisfies a specification, in this case, valid JSON.
\jsonbench{} is implemented with a batching server for local LLM inference, combining simultaneous LLM call inputs into a large tensor, which PyTorch~\cite{pytorch} can compute more efficiently than it can each call individually. This deliberately trades off a small amount latency for running time, as the initial LLM call doesn't return until the whole initial batch is done, which is slower than evaluating sequentially for a single initial call. As a result, the latency numbers for \jsonbench{} are inverted from what one might expect: sequential evaluation has the lowest latency, while the manually parallelized implementation has the highest. Batch sizes can be tuned to alter the latency-running time tradeoff.

\paragraph{Bounding Variance for External APIs.}
External calls, particularly to LLMs, return nondeterministic outputs and therefore lead to high variance in execution time. 
To limit timing variance across different runs of the same benchmark and perform a fair comparison against different implementations, we implemented a record-replay system that has two modes of execution: (1) it records calls to external APIs (OpenAI and Wikipedia) and the precise time each of their output chunks is received, and (2) it replays these recordings as if they were live calls.
We used this system to record three runs of each benchmark, that we then replayed to compare all implementations.
The results of the record-replay system are valid assuming that the timing of external APIs does not depend on the number of calls made to them in parallel, which is a reasonable assumption given the high scalability of these APIs and the small number of calls we make in our benchmarks.

\paragraph{Measurements.}
We use two metrics to evaluate the performance of different configurations: ``latency'' is defined as the time between the start of execution and the first byte of output, while ``running time'' is the time until the output is complete. 
We report latency and running time averaged over three runs of each program.

\paragraph{Setup.}
\jsonbench{} experiments were run on a server with two Intel Xeon Gold 6148 2.40GHz 20-core/40-thread CPUs and 754 GiB of RAM; all others were run on a laptop with an Intel Core i7-8550U CPU @ 1.80GHz and 16 GiB of RAM. For remote LLM calls, we used OpenAI's GPT-3.5-Turbo. For \jsonbench{}'s local LLM, we used GPT-Neo-125M~\cite{gpt-neo} via HuggingFace Transformers~\cite{hf-transformers}. For Apache Flink~\cite{flink}, we used version 1.18.0, with Java 1.17. For PaSh~\cite{pash,pash-osdi-2022}, we used the osdi22-ae branch on GitHub, with width 8. 
For GNU bash, we used version 5.1.16(1)-release. For SGLang, we used version 0.4.6.post5.

\begin{table}
\caption{Latency and running time (seconds) for different implementations of each benchmark program. \toolname{} is our approach. ``--'' means that the framework or language cannot express the program. }
\footnotesize
\begin{tabular}{l|lrrrrrrrrr}
\toprule
& \textbf{Program} & \textbf{Python} & \textbf{bash} & \textbf{PaSh} & \textbf{Flink} & \textbf{SGL-V} & \textbf{SGL-F} & \textbf{CBV} & \textbf{\toolname{}} & \textbf{Manual} \\
\midrule
\textbf{Latency}
& \loopbench{}     &  $7.37$ &                    $1.41$ &                     $7.79$ &                    $4.13$* & $7.37$                 & $7.37$                 & $1.31$ & $1.32$ & $1.27$ \\
& \pipelinebench{} &  $8.15$ &                    $0.73$ &                     $1.26$ & \multicolumn{1}{c}{--}     & $8.14$                 & $2.81$                 & $2.95$ & $0.64$ & $0.61$ \\
& \toolusebench{}  &  $3.34$ &                    $1.37$ &                    $14.78$ & \multicolumn{1}{c}{--}     & \multicolumn{1}{c}{--} & \multicolumn{1}{c}{--} & $3.35$ & $0.57$ & $0.56$ \\
& \totbench{}      &  $6.94$ &                    $8.38$ &                     $3.81$ & \multicolumn{1}{c}{--}     & $6.81$                 & $2.62$                 & $6.92$ & $2.63$ & $2.55$ \\
& \jsonbench{}     &  $0.06$ & \multicolumn{1}{c}{--}    & \multicolumn{1}{c}{--}     & \multicolumn{1}{c}{--}     & \multicolumn{1}{c}{--} & \multicolumn{1}{c}{--} & $0.06$ & $0.17$ & $0.39$ \\
\midrule
\textbf{Running}
& \loopbench{}     & $50.63$ &                    $46.27$ &                    $12.03$ &                    $9.23$* &                $50.58$ &                 $8.21$ & $49.90$ &  $9.23$ &  $8.05$ \\
\textbf{Time}
& \pipelinebench{} &  $8.15$ &                     $8.51$ &                     $3.03$ & \multicolumn{1}{c}{--}     &                 $8.14$ &                 $2.81$ &  $3.54$ &  $3.30$ &  $2.79$ \\
& \toolusebench{}  & $44.53$ &                    $67.13$ &                    $19.46$ & \multicolumn{1}{c}{--}     & \multicolumn{1}{c}{--} & \multicolumn{1}{c}{--} & $45.33$ & $10.47$ & $10.34$ \\
& \totbench{}      & $49.29$ &                    $66.24$ &                    $19.19$ & \multicolumn{1}{c}{--}     &                $47.98$ &                 $8.13$ & $49.23$ &  $7.99$ &  $7.29$ \\
& \jsonbench{}     & $23.97$ & \multicolumn{1}{c}{--}     & \multicolumn{1}{c}{--}     & \multicolumn{1}{c}{--}     & \multicolumn{1}{c}{--} & \multicolumn{1}{c}{--} & $24.63$ & $15.57$ & $15.30$ \\
\bottomrule
\end{tabular}
\vspace{0.25em} \\
{\scriptsize \ \ *There are two Flink implementations, one with better latency, and one with better running time; the best of each is reported.}\hfill
\label{fig:tab-timing-new}
\end{table}

\subsection{RQ 1: Comparison to a Practitioner-Style Python Implementation}
\label{sec:rq1}
Python is the \emph{de facto} standard language for machine learning, and a quick survey of notable LLM script implementations (Tree-of-Thoughts~\cite{tot-impl}, ReAct~\cite{react-impl}, Toolformer~\cite{toolformer-impl}, Synchromesh~\cite{synchromesh-impl}) shows that practitioners often use neither parallelization nor streaming primitives in Python, opting instead for sequential, blocking code. We implemented each benchmark program in Python, without parallelization, and using blocking calls to LLMs. \toolname{} achieves better latency and running time across the board, by a factor of $2.6\times$ to $12.7\times$ (latency) and $1.5\times$ to $6.2\times$ (running time), with the exception of \jsonbench{}'s latency, because the Python implementation does not exploit batching.
For \totbench, we also compared to the authors' original implementation~\cite{tot-impl}, finding that it has similar performance as our Python implementation: $56.16 \pm 10.57$ (across 9 runs) for the authors', versus $49.29 \pm 4.88$ for ours; 
the discrepancies are due to the fact that our prompts are slightly different to maintain consistency across all different implementations (Shell, PaSh, \toolname{}, and Manual).
In comparison, the running time of \toolname{} is $7.99 \pm 1.00$.
Thus, \toolname{} yields large speedups in both latency and running time over the \emph{de facto} standard approach without explicit parallelization or streaming.

\subsection{RQ 2: Comparison to Existing Frameworks for Parallelization and Streaming}

\paragraph{Baselines}
We compare the performance (latency and running time) of \toolname{} against four widely used frameworks
for parallelization and/or streaming (bash, PaSh~\cite{pash}, Apache Flink~\cite{flink}, and SGLang~\cite{sglang}).
Shell scripts are widely used to compose commands and tools to perform complex tasks, and offers some streaming and parallelization out of the box through pipelining.
PaSh~\cite{pash,pash-osdi-2022} is a state-of-the-art automatic parallelization system for the shell that does not require any modification on the scripts.
The \jsonbench{} benchmark cannot be implemented directly as a shell script (bash, PaSh) because it requires tight integration with a Python LLM library with non-trivial data structures.
Apache Flink~\cite{flink} is a state-of-the-art distributed stream processing framework that can parallelize and stream applications as long as they are written using Flink's API; only the \loopbench{} benchmark can be implemented in Flink in a way that exposes some parallelization.
SGLang~\cite{sglang} is a framework for composing LLM calls, and offers a ``frontend'' language as an embedded DSL in Python for writing LLM scripts. SGLang cannot implement the \toolusebench{} and \jsonbench{} benchmarks, as it does not support tool calls or logit access.
Table~\ref{fig:tab-timing-new} shows the results for all benchmarks and configurations.

\paragraph{Bash}
For \loopbench{}, \toolusebench{}, and \pipelinebench{} the shell achieves latency that is comparable to \toolname{} ($1.41$ vs $1.32$, $1.37$ vs $0.57$, $0.73$ vs $0.64$, resp.). For \totbench{}, \toolname{} achieves lower latency ($3.2\times$, resp.). In all cases, \toolname{} has much lower running time (between $2.6\times$ and $8.3\times$).

The shell is designed for line-by-line streaming through the use of Unix pipes and therefore can achieve similar latency to \toolname{} for many use-cases. 
However, for \totbench{}, opportunities for parallelization (independent LLM calls to rank states) occur before the first program output (which happens after sorting and therefore cannot produce output until having received all its input)---by parallelizing calls before the sort, \toolname{} can achieve much better latency.

\paragraph{PaSh}
\toolname{} achieves better latency than PaSh on all benchmarks ($5.9\times$ for \loopbench{}, $2.0\times$ for \pipelinebench{}, $26.2\times$ for \toolusebench{}, $1.4\times$ for \totbench{}).
PaSh is slightly faster on \pipelinebench{} ($3.03$ vs $3.30$), and slower on the others ($1.3\times$ for \loopbench{}, $1.9\times$ for \toolusebench{}, $2.4\times$ for \totbench{}).

While PaSh supports streaming similarly to the shell, it is primarily designed for higher throughput on large-batch workloads that take minutes or hours to execute. To achieve this, it processes data in microbatches (1~MB), which improves throughput for large data sizes but increases latency when data is small (in the order of hundreds and thousands of characters for these experiments).
\toolname{} achieves better total running time for \toolusebench{} and \totbench{} compared to PaSh because it can exploit more fine-grained parallelization opportunities, for example, for \toolusebench{}, PaSh can only parallelize across different facts, while \toolname{} also parallelizes tool calls for each fact.

\paragraph{Flink}
Flink is only able to achieve nontrivial streaming and parallelization in the  \loopbench{} benchmark, through two structurally different implementations with different tradeoffs (streams-of-lists, lists-of-streams). The stream-of-lists implementation achieves the same total time as \toolname{} ($9.23$s), but its latency is worse ($6.94$s vs $1.32$s). The list-of-streams implementation has $9.61$s total time and better latency, however still worse than \toolname{} ($4.13$s vs $1.32$s).

Flink uses the ``dataflow model'' of streaming computation~\cite{akidau2015dataflow}: programs are dataflow graphs, where nodes represent operators and edges represent streams of values. For Flink, these nodes are Java functions, optionally with state, and values in streams are Java objects. 
Flink's model is capable of exposing parallelization across nodes in the graph and for shards of the same node if its state can be partitioned.
Crucially, Flink does not support streams of streams, making it impossible to ``fully'' stream the output of \loopbench{}, which invokes a streaming LLM that returns a stream of cities, where another streaming LLM call for each city gets excursions in that city.
Flink can have either: a stream of lists of characters (streams-of-lists) being able to stream output across cities but not stream the output of the excursions of a single city;
or a list of streams of characters (lists-of-streams) forfeiting streaming of the cities. In either case, one LLM call must be blocking, leading to high latency, whereas \toolname{} can handle streams-of-streams.
\pipelinebench{} requires spawning external processes, which is not supported by Flink since it would complicate fault tolerance in the context of distributed execution.
\toolusebench{}, \totbench{}, and \jsonbench{} are inexpressible (except in a trivial way) because they have top-level do-until loops: such a control flow operator in Flink has to be implemented as a single node sequential Java implementation.

\paragraph{SGLang} The SGLang frontend language allows a limited form of automatic parallelization: LLM calls return futures, which are then implicitly waited for when used. Thus we implemented our programs in two ways: \emph{verbatim} (SGL-V), where the sequential Python implementation was ported with minimal changes, which essentially results in a sequential implementation, and \emph{futures} (SGL-F), where programs are rewritten to achieve parallelism in the presence of futures, essentially resulting in the manual parallel implementation. The SGLang frontend does not support streaming.

For \totbench{}, the latency and running time of ``verbatim'' is similar to the other sequential approaches, ``Python'' and ``CBV'', while the latency and running time of ``futures'' is similar to the other parallel approaches, ``Manual'' and ``Opal''. Likewise for \loopbench{} and \pipelinebench{}, except their ``futures'' versions have poor latency due to SGLang's lack of streaming.

\subsection{RQ 3: Comparison to a Hand-Optimized Implementation}
We compare \toolname{} against manually-optimized implementations in Rust (``Manual''), 
using lightweight threads and queues to communicate partial results between them and achieve the maximum possible parallelism and streaming.
\toolname{} incurs low overheads with respect to latency ($1.4\%$ to $4.8\%$) and running time ($1.3\%$ to $18.5\%$) compared to Manual. Manual has higher latency on \jsonbench{}, due to batching (see Section~\ref{sec:rq1}).
\toolname{} can automatically exploit all streaming and parallelization opportunities, and the small differences in performance are due to \toolname{}'s interpreter overhead. The performance differences are so small because these scripts are almost entirely bottlenecked by their external calls.
The Manual implementations are significantly more complex, requiring spawning of asynchronous tasks communicating via queues.
In contrast, \toolname{}'s performance benefits come without sacrificing programmability: the developer can describe their computation without any concern for parallelism and streaming, obtaining these benefits for free.

\subsection{RQ 4: Comparison to an Ablation Using Sequential Evaluation}

We evaluate the benefits of opportunistic evaluation by performing an ablation study using two configurations of \toolname{}, one with opportunistic evaluation strategy (``\toolname{}'') and one with a call-by-value (CBV) evaluation strategy (``CBV''). 
Opportunistic evaluation yields much better running time than CBV on all programs (from $1.6\times$ to $6.2\times$) except \pipelinebench{}, and better latency (from $2.6\times$ to $5.9\times$) on all programs except \loopbench{}, where they are indistinguishable, and \jsonbench{}, where batching of LLM inference trades running time off slightly against latency (see. Section~\ref{sec:rq1}).
\toolname{} has better performance because CBV evaluation does not exploit parallelization, executing function calls in sequential order, whereas opportunistic evaluation does fair scheduling of all function calls.
The latency of \loopbench{} and running time of \pipelinebench{} are exceptions, in that \toolname{}-CBV gets the benefits of \toolname{} (as compared to Python and other baselines). This is because, even with CBV evaluation, streaming ChatGPT call immediately return lists with holes in them, and the holes are filled in in the background as characters are received. CBV only blocks when the missing list values are needed, which happens not to harm latency in \loopbench{} and running time in \pipelinebench{}. CBV gets better latency on \jsonbench{}, because it does not exploit batching.

\section{Related Work}

\paragraph{DSLs for LLM Programming.}
Recent work has developed several Python DSLs for different aspects of LLM programming. Most of these techniques are orthogonal to our language and execution model and could be integrated in future work.
LangChain~\cite{langchain} allows users to build ``chains'' of LLMs and other components, including arbitrary functions, analogous to pipelines in shell scripts. However, it is unable to express more complex use-cases, such as Tree-of-Thoughts~\cite{tree-of-thoughts}, which was recently added to LangChain as a primitive.
LlamaIndex~\cite{llamaindex} is a library to help users structure data, such as internal company documents, for provision to LLMs for use in tasks like retrieval-augmented generation~\cite{lewis2020retrieval}, a form of tool-use.
Guidance~\cite{guidance}, LMQL~\cite{lmql}, and SGLang~\cite{sglang}'s frontend are Python embedded DSLs that facilitate advanced prompting, interleaving interactions with a LLMs with arbitrary Python constructs. However, they are unable to express more complex use-cases. In addition to its frontend, SGLang has a model-serving backend, which is a drop-in replacement for OpenAI and thus can be used with \toolname{} with no change. SGLang's frontend supports explicit fork/join parallelism, and has ``compiler'' and ``interpreter'' modes. The compiler mode builds a data-flow graph but does not support control-flow. SGLang's interpreter mode, Guidance, and LMQL all rely on Python for control-flow, and so do not understand and cannot parallelize control-flow constructs.
Galois~\cite{saeed2023querying} proposes treating LLMs as databases and querying them with SQL. DSPy~\cite{dspy} is a language for specifying prompting strategies, along with a framework for fine-tuning models and learning prompts to optimize a dataset performance metric.

\paragraph{Parallel Computation in Haskell}
Parallel Haskell~\cite{marlow2013parallel} leverages purity for efficient parallelization. However, it requires manual user annotations for ``evaluation strategies'', and it requires that the code be truly pure. Effects, even commuting ones, must be serialized via the IOMonad~\cite{jones2001tackling}. Concurrent Haskell~\cite{concurrent-haskell} supports effects, but requires the user to manually fork and synchronize processes. Our approach is spiritually similar to using \kw{unsafePerformIO} in Haskell, however, because Haskell is lazy, it may drop external calls. Further, since our external calls are impure, many compiler optimizations become invalid. For instance, common sub-expression elimination could, for pure code, rewrite \kw{x := coin(); y := coin(); (x, y)} to \kw{x := coin(); (x, x)}. However, this is clearly unsound for nondeterministic \kw{coin}, reducing two independent coin flips to a single one. Using unsafePerformIO for such nondeterministic effects requires disabling various compiler optimizations (e.g. common sub-expression elimination, function inlining, and let-floating).

There is also a Haskell library called Haxl~\cite{haxl}, which takes a similar approach to ours, applied to parallel data access from remote, read-only APIs. It uses a monadic domain-specific language to identify a set of requests that can be made in parallel. However, computation rigidly alternates between executing the user program to collect parallel requests and waiting for all in-flight requests to return. They leave the sort of interleaving that \toolname{} does for future work. Additionally, they assume that the remote APIs are pure, which is not true for LLM scripting.

\paragraph{Automatic Parallelization for Shell Scripts}
Recent work focusing on Unix shell scripts achieves performance benefits through automatic parallelization, distribution, and out-of-order execution~\cite{pash,pash-osdi-2022,dish2023nsdi,hs2023hotos}.
While these systems achieve significant performance benefits, they focus on parallelization at a coarser granularity, 
    that of external commands that run on their own processes,
    and are not able to expose the fine-grained parallelization and streaming possible with \toolname{}.

\paragraph{Languages for Parallel Computation}
There are numerous languages that support explicit constructs for parallelism, including NESL~\cite{nesl}, Cilk~\cite{cilk}, and Chapel~\cite{chapel}, as well as C++ (via the Parallel library), Python (via async/await), and other mainstream languages. However, parallelism in all of these languages is not automatic, requiring code changes. By contrast, parallelism in \toolname{} is completely automatic, which is easier for non-expert users of scripting languages.

\paragraph{Dataflow Programming} Dataflow languages~\cite{Karp1966,Adams1968,Karp1969,Adams1969,Rodriguez1969,dennis1974dataflow,davis1982data}, including VAL~\cite{mcgraw1982val}, SISAL~\cite{sisal}, Id~\cite{arvind1992id}, pH~\cite{maessen1995semantics}, and others~\cite{dataflow-history} are known to be highly parallelizable. However, these languages are strict (except Id and pH, discussed below). Strictness loses many opportunities for parallelism, since e.g. the term \kw{(fun x: x + (llm "bar")) (llm "foo")}, can't begin to evaluate \kw{(llm "bar")} until after the (slow) call to \kw{(llm "foo")} is completed.

\paragraph{Automatically Parallel Lambda Calculi}
The most closely related work is $\lambda_S$~\cite{maessen1996s}, a core calculus for the languages Id~\cite{arvind1992id} and pH~\cite{maessen1995semantics}. It is automatically parallel, with no explicit parallelism constructs. Like $\lambda^O$, it employs a lenient evaluation strategy. Being a cyclic lambda calculus~\cite{cyclic-lambda}, it is not confluent, unlike $\lambda^O$. It has a deterministic subset, $\lambda_C$ which possesses the weaker notion of ``print-confluence''. $\lambda_S$ adds notions of mutable heap memory and synchronization barriers on top of $\lambda_C$, which lead to nondeterminism, and deprives $\lambda_S$ of print-confluence. In contrast to $\lambda^O$, $\lambda_S$ does not have a notion of external calls, and is not confluent in the presence of external-call nondeterminism. We are not aware of work applying pH or Id to scripting or machine learning applications. Our ``tasks'' are very similar to ``futures'' in Multilisp~\cite{multilisp}, where are resolved implicitly. Note that these are distinct from ``futures'' in mainstream languages like Python, which must be explicitly awaited.

\paragraph{Lenient Evaluation}
Though uncommon, there are evaluation strategies that are neither lazy nor strict. They may do more reduction than lazy evaluation, which does strictly the smallest number of reduction steps to reach a normal form, if it exists. This is desirable because it allows parallel evaluation that is not possible with lazy evaluation. A \emph{lenient} evaluation~\cite{traub1988sequential,TREMBLAY200027,TREMBLAY200043,cyclic-lambda} strategy is one that is neither lazy nor strict, but also must still reach a normal form if it exists. Our opportunistic evaluation strategy is lenient, but it is more stringent: lenient evaluation may discard statements, whereas opportunistic evaluation must step every statement. We formulate and prove this property.

\paragraph{Linear Logic}
Our language draws inspiration from Ideograph~\cite{ideograph}, which represents programs directly as linear, higher-order data-flow graphs and has a similar parallel evaluation strategy. Our ``step-call'' rule is closely related to sequent calculi for linear logic~\cite{linearlogic}, in particularly the ``adsorption'' rule from the Dyadic System $\Sigma_2$~\cite{andreoli1992logic}.

\section{Conclusion and Future Work}
In this paper, we presented an approach for automatic parallelization and streaming of scripts that make effectful external calls, using a novel \emph{opportunistic} evaluation strategy. Our work demonstrates the feasibility of creating programming languages that abstract over parallelization and streaming the way current high-level languages abstract over memory management. Important directions for future work include:
\begin{itemize}
  \item \emph{Expanding the scope of opportunistic evaluation.} \toolname{} works by executing parts of the program much earlier than they ordinarily would be. However, there are still cases where something could be evaluated early, but \toolname{} does not. For example, \kwtt{f(if b then x else y)} and \kwtt{if b then f(x) else f(y)} are intuitively equivalent, but \toolname{} only opportunistically evaluates \kw{f} in the former. In the latter, it must wait for \kwtt{b} to be a value in order to know which branch to execute. Characterizing what it means for a system to be ``more opportunistic'' and expanding the scope of what can be opportunistically evaluated would further improve performance.
  \item \emph{Improving frontend support for common language features.} Currently, \toolname{} only supports Python functions and tuples, requiring other language constructs to be encoded by the programmer (e.g. \kw{for} loops must be encoded as functional \kwtt{fold}s). Compiling a larger subset of Python down to \calcname{} would be challenging, but would significantly improve the usability for programmers.
  \item \emph{Ensuring correctness of external call sequencing.} Sequencing of external calls is achieved in \toolname{} via data-dependencies (Section~\ref{sec:epiclang:sequencing}), but the correctness of such sequencing is not checked. For example, print statements should be sequential, but in \kwtt{_ = print(stdout1, <"a">);} \kwtt{stdout2 = print(stdout1, <"a">)},  the print order is nondeterministic, since \kw{stdout1} is used twice. Adding types, and in particular a notion of linearity~\cite{linearlogic}, could improve \toolname{}'s handling of sequencing.
  \item \emph{Closing the performance gap vs manual implementations.} Though \toolname{} yields large speedups over automatically parallel and streaming baselines, it still has noticeable overhead when compared to manual parallelization and streaming in low-level languages. One promising area for improvement is the handling of streaming: Church encodings implicitly use an inefficient linked-list implementation of strings, which incurs significant overhead; developing a notion of partial values and streaming for primitives could significantly improve performance.
\end{itemize}
Addressing these challenges will significantly improve the usability of \toolname{} in practice.

\begin{acks}
We thank the anonymous reviewers for their helpful feedback.
This work was supported in part by
NSF Awards CCF-1917852, CCF-2247088, and CCF-2338777,
Amazon Research Award Fall 2023, and
Amazon/ASSET Gift for Research in Trustworthy AI.
Any opinions, findings, and conclusions or recommendations expressed in this material are those of the authors and do not necessarily reflect the views of funding entities.
\end{acks}

\section*{Data Availability Statement}
The implementation of \toolname{} \href{https://doi.org/10.5281/zenodo.16929280}{is available}~\cite{artifact} as an installable Python package and a Dockerfile. It supports writing and executing novel programs, as well as producing the numerical results in this paper. For reproducibility, we include recorded network exchanges with OpenAI (request, response, and timing) used in our experiments.

\appendix
\section{Well-Formedness}
\label{appendix:wf}

\begin{figure}[H]
    \centering
    \small
    \begin{mathpar}
        \inferrule{
            x \notin \Gamma \\
            \Gamma \vdash o \\
            \Gamma, x \vdash e \\
        }{
            \Gamma \vdash \sstmt{x}{o} \ssep e
        }
        \text{WF-e-stmt} \and
        \inferrule{
            x \in \Gamma
        }{
            \Gamma \vdash x
        }
        \text{WF-e-var} \and
        \inferrule{
            x \in \Gamma
        }{
            \Gamma \vdash x
        }
        \text{WF-o-var} \and
        \inferrule{
            x \notin \Gamma \\
            \Gamma, x \vdash e \\
        }{
            \Gamma \vdash \efun{x}{e}
        }
        \text{WF-o-fun} \and
        \inferrule{
            f, x \in \Gamma \\
        }{
            \Gamma \vdash \eapp{f}{x}
        }
        \text{WF-o-call} \and
        \inferrule{
            x_1, \ldots, x_n \in \Gamma \\
        }{
            \Gamma \vdash \etuple{x_1, \ldots, x_n}
        }
        \text{WF-o-tuple} \and
        \inferrule{
            x \in \Gamma \\
        }{
            \Gamma \vdash \eproj{i}{x}
        }
        \text{WF-o-proj} \and
        \inferrule{
        }{
            \Gamma \vdash \eprim{c}
        }
        \text{WF-o-prim} \and
        \inferrule{
        }{
            \Gamma \vdash \etask{c_f}{\bar{c}_x}
        }
        \text{WF-o-task}
    \end{mathpar}
    \caption*{The rules for syntactic well-formedness of terms in \calcname{}.}
    \label{fig:corecalculus:syntax:wf}
\end{figure}

\bibliography{refs}

\end{document}